\documentclass[letterpaper, 11pt]{article}

\usepackage{latexsym}
\usepackage{epsfig}
\usepackage{amsmath} 	
\usepackage{amssymb}
\usepackage{rotating}
\usepackage{amsthm}
\usepackage{fullpage}
\usepackage{amsthm}

\newtheorem{theorem}{Theorem}
\newtheorem{lemma}{Lemma}
\newtheorem{proposition}{Proposition}
\newtheorem{corollary}{Corollary}
\newtheorem{definition}{Definition}

\newtheorem{claim}{Claim}

\numberwithin{equation}{section}

\title{The Hilbertian Tensor Norm \\ and Entangled Two-Prover Games}

\author{Dejan D. Dukaric \\
[0.5cm]
Institute of Theoretical Computer Science, ETH Zurich, 8092 Zurich, Switzerland  \\
Institute for Theoretical Physics, ETH Zurich, 8093 Zurich, Switzerland \\  
[0.5cm]
\texttt{ddukaric@ethz.ch}
\\
[0.5cm] }

\begin{document}

\maketitle


\begin{abstract}
We study tensor norms over Banach spaces and their relations to quantum information theory, in particular their connection with two-prover games. We consider a version of the Hilbertian tensor norm $\gamma_2$ and its dual $\gamma_2^*$ that allow us to consider games with arbitrary output alphabet sizes.
We establish direct-product theorems and prove a generalized Grothendieck inequality for these tensor norms. Furthermore, we investigate the connection between the Hilbertian tensor norm and the set of quantum probability distributions, and show two applications to quantum information theory: firstly, we give an alternative proof of the perfect parallel repetition theorem for entangled XOR games; and secondly, we prove a new upper bound on the ratio between the entangled and the classical value of two-prover games.
\end{abstract}

\newpage

\tableofcontents

\section{Introduction and Motivation}

Entanglement is one of the central and most fascinating properties of quantum mechanics. The strange consequences of entangled quantum states have already puzzled Einstein, Podolsky, and Rosen \cite{EPR}, in their seminal paper of 1935 in which they raise the issue whether quantum mechanics is complete. This leads to the question if it is possible to augment quantum mechanics with additional (yet) unknown parameters, so called \emph{local hidden variables} (LHV), in order to obtain a local realistic and complete  theory. It was John Bell who gave a negative answer  to this question. He showed \cite{bellInequality} that there exist entangled quantum states and local measurements such that the resulting conditional probability distributions cannot be explained by a LHV theory. Consequently, such behaviours are called \emph{non-local}. In the last two decades, non-locality has become an extensively studied subject within quantum information theory which has applications in subjects ranging from device independent quantum key distribution \cite{nosignalingQKD, deviceIndependentCollective, deviceIndependentGeneral} over questions about the foundations of quantum mechanics \cite{limitOnNonlocality, closedSets, beyondquantum} to multi-prover games \cite{BGKW89, limitsNonlocal, paralleXOR, uniqueGamesEasy, entangledHard, noStrongParallel}. 

In a two-prover game Alice and Bob, the provers, are space-like separated and receive each a classical question from a verifier. Then, each of them sends back a classical answer to the verifier. The goal of the provers is to maximize the winning probability for the predefined game. This maximal winning probability can depend on the resources Alice and Bob share. Typically it is higher if they have non-local resources at their disposal instead of only shared randomness. In order to better understand the power and limitations of quantum non-locality it is therefore of interest to investigate the question of how big the gap between the winning probabilities with quantum and classical resources can maximally be. Note that for the case of XOR games, i.e., games for which the winning condition only depends on the XOR of Alice's and Bob's answer bit, this question has been fully answered by Tsirelson \cite{tirelsonGrothendieckearly}, who showed that there is a constant gap, independent of the input alphabet sizes.

Despite the fact that non-locality has been extensively studied there are still many important open questions. One of the questions addresses the problem of deciding whether a given conditional probability distribution can be obtained by product measurements on a quantum state. There is no efficient algorithm known which decides this problem. The current state of the art is an infinite hierarchy of semi-definite programs \cite{hierarchySDP1, hierarchySDP2} which decides whether a system is not quantum. The drawback of this approach is that the convergence rate is, in general, not known. In order to overcome the shortage of knowledge about specific properties of the quantum set we consider a relaxation of the quantum set, obtaining a larger set of conditional probability distributions. This larger set has desirable properties while still being reasonable close to the quantum set. We will show that $\gamma_2$ can be used to define such a bigger set. In addition, by considering the dual Hilbertian tensor norm, denoted by $\gamma_2^*$, we are able to make statements about the winning probability of two-prover games with Alice and Bob having quantum systems as resources. 

The first one who observed that there is a connection between tensor norms and quantum systems was Tsirelson \cite{tirelsonGrothendieckearly}. He showed that the ratio between the maximal quantum and the maximal classical value of XOR games\footnote{Note that his result also holds for the more general setting of correlation Bell inequalities.} is bounded by the Grothendieck constant $1.68 \lesssim K_G \lesssim 1.78$. Tsirelson used Grothendieck's inequality \cite{grothendieckInequality} which establishes a connection between the Hilbertian and the projective tensor norm. Together with the fact that there is a one-to-one correspondence between quantum correlations and the Hilbertian tensor norm and between classical correlations and the projective tensor norm, the constant gap between the quantum and classical winning probabilities is implied.

\paragraph*{Our contribution:} In this paper we generalize the above mentioned argument of Tsirelson. 
First, we establish a connection between arbitrary quantum systems and the Hilbertian tensor norm. In particular,  we prove that $\gamma_2$ evaluates to one for all quantum systems (see Proposition \ref{lemmaGammaQuantumState1} in Section \ref{sectionHilbertTensor}). Note however, that in contrast to the case of quantum correlations this result is not tight. 

And second, we introduce a generalized Grothendieck inequality which can be applied in a setting where Alice and Bob have several possible outputs, and is therefore an extension from XOR games to arbitrary two-prover games (see Theorem \ref{theoremDualGrothen} in Section \ref{sectionGeneralGrothen}). In Section \ref{sectionGeneralGrothen1} we provide a dual tensor and a matrix version of this generalized Grothendieck inequality. 

Combining these two results allows us to upper bound the ratio between the maximal quantum and the maximal classical value of arbitrary two-prover games (see Theorem \ref{theoremBellupperNew} in Section \ref{sectionApp2}) and to improve the best known upper bound given in \cite{commComplexNonSig} by a square root factor.

In Section \ref{sectionDirectProduct}, we prove a direct-product theorem for the dual Hilbertian tensor norm $\gamma_2^*$ (see Theorem \ref{lemmaProductTheoremGenNorm1}). This generalizes work of \cite{directProductTheorems} and enables us, together with a new tight characterization of the entangled winning probability for XOR games by means of $\gamma_2^*$, to derive an alternative proof of the perfect parallel repetition theorem for entangled XOR games (see Theorem \ref{theoremXORGames} in Section \ref{sectionApp1}).

\paragraph*{Related work:} Using tools from operator space theory, Junge and Palazuelos \cite{largeViolation} study large violations of Bell inequalities.  In order to prove that their results are almost tight, they also provide results corresponding to our Theorem \ref{theoremGenerGrothenTensor} (the generalized Grothendieck inequality) and Theorem \ref{theoremBellupperNew} (the upper bound on the entangled value of two-prover games). Note that their result is more general as it holds for Bell inequalities as well. This line of research is a continuation of \cite{operatorSpace1,operatorSpace2} where it is shown that operator space theory is a natural framework to study arbitrary Bell inequalities. The authors of \cite{nearOptimalBell} improve the work of \cite{largeViolation} by providing explicit two-prover games in order to establish near optimal lower bounds on the ratio between the quantum and classical value of Bell inequalities.

Grothendieck's inequality has been generalized in different ways before. The latest generalization can be found in \cite{generalGrothen} where references to other previous generalizations \cite{genProofGrothen, bellInGrothen, quadGrothGraphs} are provided.
Grothendieck's inequality and, in particular, the tensor norm $\gamma_2$ and its dual $\gamma_2^*$, have not only applications in quantum information theory but are also used to prove lower bounds in communication complexity \cite{complexMeasure, lowerBoundsComm, directProductTheorems}. Furthermore Grothendieck's inequality serves as an inspiration to derive new semi-definite programs which can be used to approximate computationally hard problems \cite{ApproxCutNorm, maxQuadratic}.

\section{Preliminaries and Notation}
\label{sectionPreNot}

\subsection{Two-Prover Games}
\label{sectionTwoProver}

In a \emph{classical one-round two-prover cooperative game of incomplete information} \cite{BGKW89} two classical and spatially separated provers, usually called Alice and Bob, try to win a game by interacting with a verifier. The two provers can agree on a strategy before the game. During the game the two provers are not allowed to communicate. The messages which are exchanged by the verifier and the two provers are classical bit strings. Let $
\pi: \mathcal{X} \times \mathcal{Y} \rightarrow [0,1]$ be a probability distribution known by the verifier and the two provers. The verifier selects $x \in \mathcal{X}$ and $y \in \mathcal{Y}$ according to the probability distribution $\pi$ and sends the value $x$ to Alice and $y$ to Bob. Alice and Bob send to the verifier the values $s_A(x) = a \in \mathcal{A}$ and $s_B(y) = b \in \mathcal{B}$ where we call the pair $(s_A, s_B)$ a strategy for the game. Note that it is sufficient to consider deterministic strategies only as the optimal (shared) randomness can be selected in advance. The provers win the game $G = (\pi, V)$ if the publicly known predicate $V : \mathcal{A} \times \mathcal{B} \times \mathcal{X} \times \mathcal{Y} \rightarrow \{0, 1\}$ evaluates to $1$ for the four-tuple $(a,b,x,y)$. We consider two classes of games:
\begin{definition}
Let $G = (\pi,V)$ be a game. Then
\begin{itemize}
\item $G$ is called a \emph{unique} game if there exist permutations $\sigma_{x,y}$ for all inputs $x \in \mathcal{X}$ and $y \in \mathcal{Y}$ such that $V(a,b,x,y) = 1$ if and only if $b = \sigma_{x,y}(a)$.
\item $G$ is called an \emph{XOR} game if $\mathcal{A} = \mathcal{B} = \{ 0,1 \}$ with $V(0,0,x,y) = V(1,1,x,y)$ and  $V(0,1,x,y) = V(1,0,x,y)$ for all inputs $x \in \mathcal{X}$ and $y \in \mathcal{Y}$, i.e., the predicate $V$ depends only on the XOR of the answers $a$ and $b$.
\end{itemize}
\end{definition}

The classical value of the protocol $s_A : \mathcal{X} \rightarrow \mathcal{A}$ and $s_B : \mathcal{Y} \rightarrow \mathcal{B}$ is defined by
\[
 \sum_{x \in \mathcal{X}, y \in \mathcal{Y}} \pi(x,y) \cdot V(s_A(x), s_B(y), x, y) ~.
\]
The \emph{classical value} of a game, denoted by $\omega(G)$, is defined as the maximal value that can be achieved by any two strategies $s_A$ and $s_B$ for a given game $G = (\pi, V)$, i.e.,
\[
\omega(G) := \max_{s_A, s_B} \sum_{x \in \mathcal{X}, y \in \mathcal{Y}} \pi(x,y) \cdot V(s_A(x), s_B(y), x, y) \ .
\]

We can give the two provers more power by allowing them to share entangled quantum states. Alice and Bob can then select a measurement depending on their inputs $x$ and $y$, respectively, and measure the entangled state $| \Psi \rangle$, obtaining measurement results $a$ and $b$, respectively. The \emph{entangled value} of a game $G = (\pi, V)$, denoted by $\omega^*(G)$, is defined as \cite{uniqueGamesEasy}
\[
\omega^*(G) := \lim_{d \rightarrow \infty} \max_{ | \Psi \rangle \in \mathbb{C}^d \otimes \mathbb{C}^d} \max_{M_{x}^a, N_{y}^b} \sum_{x,y} \pi(x,y) \sum_{a,b} V(a,b,x,y) \cdot \langle \Psi | M_x^a \otimes N_y^b | \Psi \rangle ~,
\]
with projective measurements $\{ M_x^a\}_{1 \leq a \leq |\mathcal{A}|}$, for $1 \leq x \leq |\mathcal{X}|$, and $\{ N_y^b\}_{1 \leq b \leq |\mathcal{B}|}$, for $1 \leq y \leq |\mathcal{Y}|$. It is clear that $\omega^*(G) \geq \omega(G)$ for all games $G$.

\subsection{Parallel Repetition of Two-Prover Games}
\label{sectionParalleRep}

A game $G = (\pi, V)$ can be repeated $N$ times
independently. Either the game is repeated
\emph{sequentially}, i.e., a full round is completed before a new round is
started, or in \emph{parallel}. In the latter case, $N$ mutually independent pairs of
inputs $(x_i, y_i)$ are chosen according to the distribution $\pi$ and sent to the provers. The provers
then compute outputs $(a_1, \ldots, a_N)$
and $(b_1, \ldots, b_N)$,
respectively. Finally, the predicate $V$ is evaluated for all tuples
$(a_i, b_i, x_i, y_i)$ separately. This $N$-fold repetition of a game $G$ can be seen as a new game, denoted by $G^{\odot N}$,  where this new game is only won if \emph{all} $N$ rounds are won.

For sequential composition, this probability is
obviously equal to the probability of winning a single game taken to the power of $N$. However, for parallel
composition the problem gets more involved as it is generally not
true that $\omega(G^{\odot N})$ is equal to $\omega(G)^N$, as
shown in~\cite{F89}. The same is true for entangled games, i.e., there exist games such that $\omega^*(G^{\odot N}) > \omega^{*}(G)^N$ \cite{noStrongParallel}. Note that $\omega(G^{\odot N}) \geq \omega(G)^N$ and $\omega^*(G^{\odot N}) \geq \omega^{*}(G)^N$ is obviously true for all games $G$. Nevertheless, it can be shown that the quantity
$\omega(G^{\odot N})$ decreases exponentially fast in $N$. A first
proof of this fact, also known as the \emph{Parallel Repetition
 Theorem}, has been given in \cite{R98}.
Raz's proof has been simplified in \cite{H07} and extended to the case of provers using arbitrary non-signalling resources.

No such parallel repetition result is known for entangled games. However, for the special case of entangled XOR games there holds a perfect parallel repetition theorem \cite{paralleXOR}. Recently the authors of \cite{uniqueGamesEasy} have shown that there is a parallel repetition theorem for entangled unique games as well. Quantitatively, it is known that if 
$G = (\pi, V)$ is a two-prover game then, for all $N \geq 1$, it holds that
\begin{itemize}
 \item \cite{paralleXOR}  if $G$ is an XOR game, then $\omega^*(G^{\odot N}) = \omega^*(G)^N$  ~,
 \item \cite{uniqueGamesEasy} if $G$ is a unique game, then $\omega^*(G^{\odot N}) \leq \left(1 - \frac{(1 - \omega^*(G))^2}{16} \right)^N$~.
\end{itemize}

\subsection{Banach Spaces}
\label{sectionLoalBnachSp}

Let $\| \cdot \|_X$ be a norm on the real finite-dimensional vector space $\mathbb{R}^n$. Then the tuple $X := (\mathbb{R}^n, \| \cdot \|_X)$ is called a \emph{Banach space}. The \emph{dual space} of $\mathbb{R}^n$, denoted by $(\mathbb{R}^n)^*$, is the vector space of all linear functionals from the vector space $\mathbb{R}^n$ to the real numbers. We write $\langle G, P \rangle \in \mathbb{R}$ for the application of the linear functional $G: \mathbb{R}^n \rightarrow \mathbb{R}$ on the element $P \in \mathbb{R}^n$. Note that this is just the usual inner product of real vectors. The corresponding \emph{dual norm} is then defined by
\begin{equation}
\label{eqnDefDualNorm}
\| G \|_{X^*} := \sup_{ P \in \mathbb{R}^n} \{ | \langle G, P \rangle | ~:~ \| P \|_X \leq 1 \} ~,
\end{equation}
and the dual Banach space is given by $X^* := ((\mathbb{R}^n)^*, \| \cdot \|_{X^*})$.

We write $\langle f_i, P \rangle$, where $f_i \in (\mathbb{R}^n)^* \cong \mathbb{R}^n$ is the all-zero vector with a one at position $i$, to access the $i$'th entry of the vector $P \in \mathbb{R}^n$.  And similarly, if $G \in (\mathbb{R}^n)^* \cong \mathbb{R}^n$ we use $\langle G, e_i \rangle$, where $e_i \in \mathbb{R}^n$ is the all-zero vector with a one at position $i$, to access the $i$'th entry of $G$. The inner product $\langle G, P \rangle$ can therefore also be written as
\begin{equation}
\label{eqnAltRepInnerProd}
\langle G, P \rangle = \sum_{i=1}^n \langle G, e_i \rangle \cdot \langle f_i, P \rangle ~.
\end{equation}

In particular, we consider the Banach space
\[
\ell_{\infty}^{|\mathcal{X}|}(\ell_{1}^{|\mathcal{A}|}) := (\mathbb{R}^{|\mathcal{X}|} \otimes \mathbb{R}^{|\mathcal{A}|}, \| \cdot \|_{\infty(1)}) ~,
\]
where the $\infty(1)$-norm is defined as
\[
\| P_A \|_{\infty(1)} := \max_{x \in \mathcal{X} } \sum_{a = 1}^{|\mathcal{A}|} | \langle f_x \otimes f_a,  P_A \rangle |  ~,
\]
for $P_A \in \mathbb{R}^{|\mathcal{X}|} \otimes \mathbb{R}^{|\mathcal{A}|}$. We will also use the notation $f_{x,a} := f_x \otimes f_a$. See also Section \ref{sectionConnProverTensor} which gives an interpretation of the expression $\langle f_x \otimes f_a,  P_A \rangle$ in the context of two-prover games. The dual space is given by $(\ell_{\infty}^{|\mathcal{X}|}(\ell_{1}^{|\mathcal{A}|}))^* \cong \ell_{1}^{|\mathcal{X}|}(\ell_{\infty}^{|\mathcal{A}|}) := (\mathbb{R}^{|\mathcal{X}|} \otimes \mathbb{R}^{|\mathcal{A}|}, \| \cdot \|_{1(\infty)} )$ with
\[
\| G_A \|_{1(\infty)} \equiv \| G_A \|_{\infty(1)^*}  := \sum_{x = 1}^{|\mathcal{X}|} \max_{a  \in \mathcal{A}} | \langle G_A , e_x \otimes e_a \rangle | ~,
\]
for $G_A \in (\mathbb{R}^{|\mathcal{X}|} \otimes \mathbb{R}^{|\mathcal{A}|})^* \cong \mathbb{R}^{|\mathcal{X}|} \otimes \mathbb{R}^{|\mathcal{A}|}$. It is easy to verify that
\[
\| G_A \|_{1(\infty)} = \sup_{P_A} \{ | \langle G_A, P_A \rangle | ~:~ \| P_A \|_{\infty(1)} \leq 1 \} ~,
\]
and therefore, the $1(\infty)$-norm is indeed the dual of the $\infty(1)$-norm.

Note that for $| \mathcal{A} | = 1$ we recover the Banach space $\ell_{\infty}^{|\mathcal{X}|} := (\mathbb{R}^{|\mathcal{X}|}, \| \cdot \|_1)$ where $\| P_A \|_1 := \sum_{x=1}^{|\mathcal{X}|} | \langle f_x, P_A \rangle |$, and for $| \mathcal{X} | = 1$ the Banach space $\ell_{1}^{|\mathcal{A}|} :=  (\mathbb{R}^{|\mathcal{A}|}, \| \cdot \|_{\infty})$ where $\| P_A \|_{\infty} := \max_{a \in \mathcal{A}} | \langle f_a, P_A \rangle |$.

We will use the symbols $P_A$, $P_B$, and $P$ for elements in $\ell_{\infty}^{|\mathcal{X}|}(\ell_{1}^{|\mathcal{A}|})$, $\ell_{\infty}^{|\mathcal{Y}|}(\ell_{1}^{|\mathcal{B}|})$, and $\ell_{\infty}^{|\mathcal{X}|}(\ell_{1}^{|\mathcal{A}|}) \otimes \ell_{\infty}^{|\mathcal{Y}|}(\ell_{1}^{|\mathcal{B}|})$, respectively, as they will represent conditional probability distributions. The symbols $G_A$, $G_B$, and $G$ are used for elements in $\ell_{1}^{|\mathcal{X}|}(\ell_{\infty}^{|\mathcal{A}|})$, $\ell_{1}^{|\mathcal{Y}|}(\ell_{\infty}^{|\mathcal{B}|})$, and $\ell_{1}^{|\mathcal{X}|}(\ell_{\infty}^{|\mathcal{A}|}) \otimes \ell_{1}^{|\mathcal{Y}|}(\ell_{\infty}^{|\mathcal{B}|})$, respectively. In this case they will represent (two-prover) games. Using this convention the expressions should be easier to read as we do not always have to explicitly mention the Banach space we are working on.

\section{Connection Between Tensor Norms and Two-Prover Games}
\label{sectionConnProverTensor}

A tensor norm is a function which maps elements from tensor product spaces $X \otimes Y$, where $X$ and $Y$ are Banach spaces, to the non-negative real numbers. Furthermore, a tensor norm inherits all properties of a regular norm and therefore fulfils the three norm-defining conditions given in Appendix \ref{sectionNotation}. A formal definition of tensor norms is given in Appendix \ref{sectionIntroTensorNorm}. 
In our particular case, we will consider the following four tensor norms (see Appendix \ref{sectionFourTensorNorms} for definitions):
\begin{itemize}
\item $\varepsilon : \ell_{1}^{|\mathcal{X}|}(\ell_{\infty}^{|\mathcal{A}|}) \otimes \ell_{1}^{|\mathcal{Y}|}(\ell_{\infty}^{|\mathcal{B}|}) \rightarrow \mathbb{R}_0^+$ (Injective Tensor Norm) ;
\item $\pi: \ell_{\infty}^{|\mathcal{X}|}(\ell_{1}^{|\mathcal{A}|}) \otimes \ell_{\infty}^{|\mathcal{Y}|}(\ell_{1}^{|\mathcal{B}|}) \rightarrow \mathbb{R}_0^+$ (Projective Tensor Norm) ;
\item $\gamma_2 : \ell_{\infty}^{|\mathcal{X}|}(\ell_{1}^{|\mathcal{A}|}) \otimes \ell_{\infty}^{|\mathcal{Y}|}(\ell_{1}^{|\mathcal{B}|}) \rightarrow \mathbb{R}_0^+$ (Hilbertian Tensor Norm) ;
\item $\gamma_2^* : \ell_{1}^{|\mathcal{X}|}(\ell_{\infty}^{|\mathcal{A}|}) \otimes \ell_{1}^{|\mathcal{Y}|}(\ell_{\infty}^{|\mathcal{B}|}) \rightarrow \mathbb{R}_0^+$ (Dual Hilbertian Tensor Norm) .
\end{itemize}
In the following we will show how one can represent a conditional probability distribution by a tensor $P \in \ell_{\infty}^{|\mathcal{X}|}(\ell_{1}^{|\mathcal{A}|}) \otimes \ell_{\infty}^{|\mathcal{Y}|}(\ell_{1}^{|\mathcal{B}|})$. This will allow us to see the projective tensor norm as a map from conditional probability distributions to non-negative real numbers. On the other hand, we will show that the tensor $G \in \ell_{1}^{|\mathcal{X}|}(\ell_{\infty}^{|\mathcal{A}|}) \otimes \ell_{1}^{|\mathcal{Y}|}(\ell_{\infty}^{|\mathcal{B}|})$ can be interpreted as a two-prover game and therefore the injective tensor norm assigns a non-negative real number to each game. We will see that this number is actually the \emph{classical} value of a two-prover game.

Let us first give an interpretation of the term $\langle f_x \otimes f_a,  P_A \rangle$, with $P_A \in \ell_{\infty}^{|\mathcal{X}|}(\ell_{1}^{|\mathcal{A}|})$, which will then lead to an explanation of the connection between tensor norms and two-prover games. First, let $s_A : \mathcal{X} \rightarrow \mathcal{A}$ be Alice's strategy. Such a strategy can always be represented by a conditional probability distribution $P_{A | X}$, with probabilities $P_{A | X}(a,x)$, output $a$, and input $x$. Setting 
\[
\langle f_x \otimes f_a,  P_A \rangle := P_{A | X}(a,x) ~,
\]
for all $x \in \mathcal{X}, a \in \mathcal{A}$, defines another representation of the conditional probability distribution $P_{A | X}$. Therefore, any (possibly probabilistic) strategy of Alice can conveniently be represented by a tensor $P_A \in \ell_{\infty}^{|\mathcal{X}|}(\ell_{1}^{|\mathcal{A}|})$. And similarly for Bob's strategy $s_B : \mathcal{Y} \rightarrow \mathcal{B}$ which can be represented by $P_B \in \ell_{\infty}^{|\mathcal{Y}|}(\ell_{1}^{|\mathcal{B}|})$. Note that in this case we have that $\| P_A \|_{\infty(1)} = \| P_B \|_{\infty(1)} = 1$. Hence, any classical strategy without shared randomness can then be represented by the product tensor $P_A \otimes P_B \in \ell_{\infty}^{|\mathcal{X}|}(\ell_{1}^{|\mathcal{A}|}) \otimes \ell_{\infty}^{|\mathcal{Y}|}(\ell_{1}^{|\mathcal{B}|})$, with
$\langle f_{x,a} \otimes f_{y,b}, P_A \otimes P_B \rangle$
representing the probability that Alice and Bob output $a$ and $b$ given they have inputs $x$ and $y$, respectively. Entangled strategies, however, can not be represent as product tensors $P_A \otimes P_B$. Instead, they will generally be represented by non-product tensors $P \in \ell_{\infty}^{|\mathcal{X}|}(\ell_{1}^{|\mathcal{A}|}) \otimes \ell_{\infty}^{|\mathcal{Y}|}(\ell_{1}^{|\mathcal{B}|})$ and given by the identification
\begin{equation}
\label{eqnrepProbTens}
\langle f_{x,a} \otimes f_{y,b}, P \rangle := P_{AB | XY}(a,b,x,y) ~,
\end{equation}
with $P_{AB | XY}(a,b,x,y)$ denoting the probability that Alice and Bob output $a$ and $b$, given the inputs $x$ and $y$, respectively.

Any two-prover game $G = (\pi, V)$ can be interpreted as an element of the tensor product space $\ell_{1}^{|\mathcal{X} |}(\ell_{\infty}^{|\mathcal{A}|}) \otimes \ell_{1}^{|\mathcal{Y} |}(\ell_{\infty}^{|\mathcal{B}|})$ by the following identification:
\begin{equation}
\label{eqnGameDefBanach}
\langle G, e_{x,a} \otimes e_{y,b} \rangle := \pi(x,y) \cdot V(a,b,x,y) ~.
\end{equation}
It will be clear from the context whether $G$ represents a pair $(\pi, V)$ or an element of a tensor product space.

The value of a protocol, which is represented as a tensor $P \in \ell_{\infty}^{|\mathcal{X}|}(\ell_{1}^{|\mathcal{A}|}) \otimes \ell_{\infty}^{|\mathcal{Y}|}(\ell_{1}^{|\mathcal{B}|})$, for a given game $G \in \ell_{1}^{|\mathcal{X} |}(\ell_{\infty}^{|\mathcal{A}|}) \otimes \ell_{1}^{|\mathcal{Y} |}(\ell_{\infty}^{|\mathcal{B}|})$, can then be computed by
\begin{eqnarray}
\label{eqnAltRepTensVal}
\langle G, P \rangle &=& \sum_{x,y} \sum_{a,b} \langle G, e_{x,a} \otimes e_{y,b} \rangle \langle f_{x,a} \otimes f_{y,b}, P \rangle \nonumber \\
&=& \sum_{x,y} \pi(x,y) \sum_{a,b} V(a,b,x,y) \cdot P_{AB | XY}(a,b,x,y) ~, 
\end{eqnarray}
where we used (\ref{eqnAltRepInnerProd}) in the first line and (\ref{eqnrepProbTens}) and (\ref{eqnGameDefBanach}) in the second line.

\subsection{Injective and Projective Tensor Norms}
\label{sectionInjecProjecNorm}

Building up on the previous section, in particular (\ref{eqnAltRepTensVal}), the \emph{classical} value of a two-prover game $G$ is then given by
\begin{equation}
\omega(G) = \sup_{P_A, P_B} | \langle G, P_A \otimes P_B \rangle | ~,
\end{equation}
where $P_A$ and $P_B$ are conditional probability distributions. As $P_A$ and $P_B$ represent strategies it follows that $\| P_A \|_{\infty(1)} = \| P_B \|_{\infty(1)} = 1$, and therefore the following upper bound on the classical value of a game is obtained:
\begin{equation}
\label{eqnValueGameREpClas}
\omega(G) \leq \sup_{P_A, P_B } \{ | \langle G, P_A \otimes P_B \rangle | ~:~ \| P_A \|_{\infty(1)} \leq 1 ~,~ \| P_A \|_{\infty(1)} \leq 1 \} ~,
\end{equation}
where $P_A$ and $P_B$ do not necessarily represent valid conditional probability distributions any more.

The right hand side of (\ref{eqnValueGameREpClas}) is usually abbreviated by $\varepsilon(G)$, i.e., it is an expression for the \emph{injective tensor norm} of $G$. Note that the dual tensor norm of $\varepsilon$ is the projective tensor norm $\pi$ (see Appendix \ref{sectionProjInjNorm}), and therefore, (\ref{eqnValueGameREpClas}) and (\ref{eqnDefDualNorm}) imply that
\[
\omega(G) \leq \varepsilon(G) = \sup_{P} \{ |  \langle G, P \rangle | ~:~ \pi(P) \leq 1 \} ~.
\]
However, as $\varepsilon(G)$ is also a lower bound on $\omega(G)$, we obtain
\begin{proposition}
\label{lemmaEqiuvClassVAr}
Let $G = (\pi,V)$ be an arbitrary two-prover game with $G \in \ell_{1}^{|\mathcal{X} |}(\ell_{\infty}^{|\mathcal{A}|}) \otimes \ell_{1}^{|\mathcal{Y} |}(\ell_{\infty}^{|\mathcal{B}|})$. Then
\[
\omega(G) = \varepsilon(G) ~.
\]
\end{proposition}

\begin{proof}
The statement follows from
\begin{eqnarray}
\varepsilon(G) &=& \sup_{P_A, P_B} \{ | \langle G, P_A \otimes P_B \rangle | ~:~ \| P_A \|_{\infty(1)} \leq 1, \| P_B \|_{\infty(1)} \leq 1 \} \nonumber \\
 &=& \sup_{P_A, P_B} \left| \sum_{x,y} \pi(x,y) \sum_{a,b} V(a,b,x,y) \cdot \langle f_{x,a}, P_A \rangle  \cdot \langle f_{y,b},  P_B \rangle \right| ~, \nonumber
\end{eqnarray}
by using (\ref{eqnAltRepTensVal}) in the second line and the supremum is over $\| P_A \|_{\infty(1)} \leq 1$ and $\| P_B \|_{\infty(1)} \leq 1$.
Thus, since $\pi(x,y) \cdot V(a,b,x,y) \geq 0$, we have that $\langle f_{x,a}, P_A \rangle \geq 0$ and $\langle f_{y,b}, P_B \rangle \geq 0$ for the optimal case. Furthermore, it is clear that the optimum is achieved when $\langle f_{x,a}, P_A \rangle$ and $\langle f_{y,b}, P_B \rangle$ are as large as possible, meaning that $\sum_{a} \langle f_{x,a}, P_A \rangle = 1$ and $\sum_{b} \langle f_{y,b}, P_B \rangle = 1$ for all $1 \leq x \leq | \mathcal{X} |$ and $1 \leq y \leq | \mathcal{Y} |$, respectively. But this implies that $P_A$ and $P_B$ correspond to valid (local probabilistic) strategies of Alice and Bob, respectively, and therefore the injective tensor norm of $G$ is the same as the classical value of the game $G$.
\end{proof}

\subsection{Hilbertian Tensor Norm and its Dual}
\label{sectionHilbertTensor}

In the previous section we have investigated the connection between tensor norms and the \emph{classical} value of two-prover games. In this section we will now establish a connection between tensor norms, in particular the dual Hilbertian tensor norm, and the \emph{entangled} value of two-prover games.

We will call $P \in \ell_{\infty}^{|\mathcal{X} |}(\ell_{1}^{|\mathcal{A}|}) \otimes \ell_{\infty}^{|\mathcal{Y} |}(\ell_{1}^{|\mathcal{B}|})$ a \emph{quantum system} if it can be obtained by product measurements on a pure quantum state $| \Psi \rangle  \in \mathcal{H}_A \otimes \mathcal{H}_B$, with $\mathcal{H}_A$ and $\mathcal{H}_B$ Hilbert spaces, i.e., there exist projective measurements $\{ M_x^a\}_{1 \leq a \leq |\mathcal{A}|}$, for $1 \leq x \leq |\mathcal{X}|$, and $\{ N_y^b\}_{1 \leq b \leq |\mathcal{B}|}$, for $1 \leq y \leq |\mathcal{Y}|$, such that (see also (\ref{eqnrepProbTens}))
\begin{equation}
\label{eqnRepQuantSys}
\langle f_{x,a} \otimes f_{y,b}, P \rangle = \langle \Psi | M_x^a \otimes N_y^b | \Psi \rangle ~.
\end{equation}
Note that it is no restriction to assume pure states and projective measurements (see also \cite{nielsenchuang}).
In Section \ref{sectionHilbertQuantum} we show that the Hilbertian tensor norm has value $1$ for all quantum systems.
\begin{proposition}
\label{lemmaGammaQuantumState1}
Let $P \in \ell_{\infty}^{|\mathcal{X} |}(\ell_{1}^{|\mathcal{A}|}) \otimes \ell_{\infty}^{|\mathcal{Y} |}(\ell_{1}^{|\mathcal{B}|})$ be a quantum system. Then $\gamma_2(P) = 1$.
\end{proposition}
Using this result, we can now upper bound the entangled value of an arbitrary two-prover game by the dual Hilbertian tensor norm.
\begin{proposition}
\label{lemmaUpperOmegaGamma}
Let $G = (\pi,V)$ be an arbitrary two-prover game with $G \in \ell_{1}^{|\mathcal{X} |}(\ell_{\infty}^{|\mathcal{A}|}) \otimes \ell_{1}^{|\mathcal{Y} |}(\ell_{\infty}^{|\mathcal{B}|})$. Then
\[
\omega^*(G) \leq \gamma_2^*(G) ~.
\]
\end{proposition}

\begin{proof}
The statement follows from
\begin{eqnarray}
 \omega^*(G) &=& \sup_{P} \left\{ \left| \sum_{x,y} \pi(x,y) \sum_{a,b} V(a,b,x,y) \cdot \langle f_{x,a} \otimes f_{y,b}, P \rangle \right| ~:~ P \textit{~quantum system} \right\} \nonumber \\
&=& \sup_{P} \{ |\langle G, P \rangle | ~:~ P \textit{~quantum system} \} \nonumber \\
&\leq& \sup_{P} \{ |\langle G, P \rangle | ~:~ \gamma_2(P) \leq 1 \} \nonumber \\
&=& \gamma_2^*(G) ~, \nonumber
\end{eqnarray}
by using (\ref{eqnRepQuantSys}) in first line, (\ref{eqnAltRepTensVal}) in the second line, Proposition \ref{lemmaGammaQuantumState1} in the third line and that $\gamma_2$ is the dual of $\gamma_2^*$ in the fourth line.
\end{proof}

\section{Generalized Grothendieck Inequality}
\label{sectionGeneralGrothen}

The previous section can be summarized by the following chain of (in)equalities which holds for any two-prover game $G$:
\[
\varepsilon(G) = \omega(G) \leq \omega^*(G) \leq \gamma_2^*(G) ~.
\]
Recall that if $G$ corresponds to a two-prover game all entries $\langle G, e_{x,a} \otimes e_{y,b} \rangle := \pi(x,y) \cdot V(a,b,x,y)$ are non-negative.
Let $G$ now be an arbitrary element of the tensor product space $\ell_{1}^{|\mathcal{X} |}(\ell_{\infty}^{|\mathcal{A}|}) \otimes \ell_{1}^{|\mathcal{Y} |}(\ell_{\infty}^{|\mathcal{B}|})$, i.e., the tensor $G$ can have negative entries as well and therefore corresponds to a general Bell inequality (see Appendix \ref{sectionIntroBell} for a short introduction to Bell inequalities). According to Lemma \ref{lemmaCrossnorm} in Appendix \ref{sectionProjInjNorm}, $\gamma_2^*(G)$ is still an upper bound on $\varepsilon(G)$. However, in Section \ref{sectionGeneralGrothen1} we prove that there is an upper bound on the maximal ratio between $\gamma_2^*(G)$ and $\varepsilon(G)$.
\begin{theorem}[Generalized Grothendieck Inequality in Dual Tensor Form]
\label{theoremDualGrothen}
For any $G \in \ell_{1}^{|\mathcal{X} |}(\ell_{\infty}^{|\mathcal{A}|}) \otimes \ell_{1}^{|\mathcal{Y} |}(\ell_{\infty}^{|\mathcal{B}|})$ it holds that
\[
\gamma_2^*(G) \leq K \cdot \sqrt{|\mathcal{A}| |\mathcal{B}|} \cdot \varepsilon(G) ~,
\]
with $K = \frac{\pi}{2 \ln(1 + \sqrt{2})}$.
\end{theorem}
The standard Grothendieck inequality \cite{grothendieckInequality} (in dual tensor form) is obtained from our generalized version by setting the output alphabet sizes to $1$, i.e., $|\mathcal{A}| = |\mathcal{B}| = 1$, and therefore
\begin{equation}
\label{eqnStandGrothTensor}
\gamma_2^*(G) \leq K_G \cdot \varepsilon(G) ~ \textit{~for~all~} G \in \ell_{1}^{| \mathcal{X} |} \otimes \ell_{1}^{| \mathcal{Y} |} ~,
\end{equation}
where $1.68 \lesssim K_G \lesssim 1.78$ is the Grothendieck constant. The exact value of $K_G$ is still unknown. Note that $K = \frac{\pi}{2 \ln(1 + \sqrt{2})} \approx 1.78$ is the best known upper bound on the Grothendieck constant $K_G$ \cite{KervinGrothendick}. The best lower bound is $K_G \geq 1.6770$ due to Reeds and Davis \cite{lowerGrothendieckDavis, lowerGrothendieckReeds}.

\section{Direct-Product Theorem}
\label{sectionDirectProduct}

We will state a direct-product theorem for the dual Hilbertian tensor norm $\gamma_2^*$ in this section. We will use this result later in the application section about parallel repetition of two-prover games. 
Let $G_{A_1B_1} \in \ell_{1}^{|\mathcal{X}_1 |}(\ell_{\infty}^{|\mathcal{A}_1|}) \otimes \ell_{1}^{|\mathcal{Y}_1 |}(\ell_{\infty}^{|\mathcal{B}_1|})$ and $G_{A_2B_2} \in \ell_{1}^{|\mathcal{X}_2 |}(\ell_{\infty}^{|\mathcal{A}_2|}) \otimes \ell_{1}^{|\mathcal{Y}_2 |}(\ell_{\infty}^{|\mathcal{B}_2|})$ be arbitrary two-prover games between Alice, Bob and the verifier. We denote by $G:= G_{A_1B_1} \odot G_{A_2B_2}$ the composition of these two games (see also Section \ref{sectionParalleRep}). Formally we have 
\begin{equation}
\label{eqnDefDecomp}
G_{A_1B_1}  \odot G_{A_2B_2}  := \sum_{i,j} (G_{A_1}^i \otimes G_{A_2}^j) \otimes (G_{B_1}^i \otimes G_{B_2}^j) ~,
\end{equation}
with $G_{A_1B_1} = \sum_i G_{A_1}^i \otimes G_{B_1}^i$ and $G_{A_2B_2} = \sum_j G_{A_2}^j \otimes G_{B_2}^j$ arbitrary decompositions.
The game $G_{A_1B_1}  \odot G_{A_2B_2}$ is therefore an element of the tensor product space $\left( \ell_{1}^{|\mathcal{X}_1 |}(\ell_{\infty}^{|\mathcal{A}_1|}) \otimes \ell_{1}^{|\mathcal{X}_2 |}(\ell_{\infty}^{|\mathcal{A}_2|}) \right) \otimes \left( \ell_{1}^{|\mathcal{Y}_1 |}(\ell_{\infty}^{|\mathcal{B}_1|}) \otimes \ell_{1}^{|\mathcal{Y}_2 |}(\ell_{\infty}^{|\mathcal{B}_2|}) \right)$.
Alice and Bob will then play these two games in parallel and try two win both rounds.

To be more explicit, we define the two round game $G$ as follows
\[
\langle G, e_{\bar{x},\bar{a}} \otimes e_{\bar{y},\bar{b}} \rangle := \pi_1(x_1,y_1) \cdot \pi_2(x_2,y_2) \cdot V_1(a_1,b_1,x_1,y_1) \cdot V_2(a_2,b_2,x_2,y_2) ~,
\]
with $e_{\bar{x},\bar{a}} := e_{x_1} \otimes e_{a_1} \otimes e_{x_2} \otimes e_{a_2}$, $e_{\bar{y},\bar{b}} := e_{y_1} \otimes e_{b_1} \otimes e_{y_2} \otimes e_{b_2}$, and $G_{A_1B_1} := (\pi_1, V_1)$ and $G_{A_2B_2} := (\pi_2, V_2)$. 
The game $G = G_{A_1B_1} \odot G_{A_2B_2}$ is therefore executed by Alice and Bob with systems $A_1A_2$ belonging to Alice and $B_1B_2$ to Bob. We say that the game is bipartite with respect to the partition $A_1A_2:B_1B_2$ between Alice and Bob, where $A_i$ corresponds to the system $\ell_{1}^{|\mathcal{X}_i |}(\ell_{\infty}^{|\mathcal{A}_i|})$ and $B_i$ to the system $\ell_{1}^{|\mathcal{Y}_i |}(\ell_{\infty}^{|\mathcal{B}_i|})$, for $i \in \{1,2\}$.

Let $P \in \left( \ell_{\infty}^{|\mathcal{X}_1 |}(\ell_{1}^{|\mathcal{A}_1|}) \otimes \ell_{\infty}^{|\mathcal{X}_2 |}(\ell_{1}^{|\mathcal{A}_2|}) \right) \otimes \left( \ell_{\infty}^{|\mathcal{Y}_1 |}(\ell_{1}^{|\mathcal{B}_1|}) \otimes \ell_{\infty}^{|\mathcal{Y}_2 |}(\ell_{1}^{|\mathcal{B}_2|}) \right)$ represent an arbitrary strategy of Alice and Bob for the two round game $G$. The winning probability of this strategy on the game $G$ is then given by $\langle G_{A_1B_1} \odot G_{A_2B_2}, P \rangle$. Note that, in general, the maximal winning probability is achieved when $P$ is a non-product strategy, i.e., $P \neq P_{A_1B_1} \odot P_{A_2B_2}$. Therefore, in general, there exist games $G_{A_1B_1}$ and $G_{A_2B_2}$ such that
\begin{eqnarray}
\sup_{P} \langle G_{A_1B_1} \odot G_{A_2B_2}, P \rangle &>& \sup_{P_{A_1B_1}, P_{A_2B_2}} \langle G_{A_1B_1} \odot G_{A_2B_2}, P_{A_1B_1} \odot P_{A_2B_2} \rangle \nonumber \\
&=&  \sup_{P_{A_1B_1}, P_{A_2B_2}} \langle G_{A_1B_1}, P_{A_1B_1} \rangle \cdot  \langle G_{A_2B_2}, P_{A_2B_2} \rangle ~, \nonumber
\end{eqnarray}
with $P, P_{A_1B_1}, P_{A_2B_2}$ either all classical or all entangled strategies.
Hence, it is in general impossible to upper bound the classical (entangled) value of the game $G_{A_1B_1} \odot G_{A_2B_2}$ by the product of the classical (entangled) values of the individual games.
On the other hand, the dual Hilbertian tensor norm has the nice property that one can actually upper bound the value of the parallel executed games by the product of the value of the single rounds.
\begin{theorem}
\label{lemmaProductTheoremGenNorm1}
Let $G_{A_1B_1} \in \ell_{1}^{|\mathcal{X}_1 |}(\ell_{\infty}^{|\mathcal{A}_1|}) \otimes \ell_{1}^{|\mathcal{Y}_1 |}(\ell_{\infty}^{|\mathcal{B}_1|})$,  $G_{A_2B_2} \in \ell_{1}^{|\mathcal{X}_2 |}(\ell_{\infty}^{|\mathcal{A}_2|}) \otimes \ell_{1}^{|\mathcal{Y}_2 |}(\ell_{\infty}^{|\mathcal{B}_2|})$  and $G_{A_1B_1} \odot G_{A_2B_2}$ be bipartite with respect to the partition $A_1A_2:B_1B_2$. Then
\[
\gamma_2^*(G_{A_1B_1} \odot G_{A_2B_2}) \leq \gamma_2^*(G_{A_1B_1}) \cdot \gamma_2^*(G_{A_2B_2}) ~.
\]
\end{theorem}
The proof is given in Section \ref{sectionDirctProductTheo}. There we also provide additional direct product results.

\section{Applications}

\subsection{Application I: Parallel Repetition of Entangled Games}
\label{sectionApp1}

In this section we will provide an alternative proof for the parallel repetition theorem for XOR games given in \cite{paralleXOR}. The proof of \cite{paralleXOR}  contains two parts. In the first part, they show that the sum of XOR games obeys a perfect product rule by using semi-definite programming (SDP) techniques and then, in a second step, they use Fourier analysis to get a perfect parallel repetition theorem for XOR games. The first part corresponds to applying the direct product result for the $\gamma_2^*$ tensor norm (see also the remark at the very end of Section \ref{sectionDirctProductTheo}) applied on a game $\tilde{G} = (\pi, \tilde{V})$, but where $\tilde{V}$ has now range $\{ -1, +1\}$ instead of $\{0,1\}$ and $\tilde{G}$ is interpreted as an element of $\ell_{1}^{|\mathcal{X} |} \otimes \ell_{1}^{|\mathcal{Y} |}$. Using Lemma \ref{lemmaQuantumCorrGamma} in Section \ref{sectionHilbertQuantum}, it is not difficult to show that $\gamma_2^*(\tilde{G}) = 2 \cdot \omega^*(\tilde{G}) - 1$. Hence, $\gamma_2^*(\tilde{G})$ is the \emph{quantum bias} of an XOR game, denoted by $\varepsilon_q(\tilde{G})$ in \cite{paralleXOR}. The second part is required because $\omega^*(\tilde{G})$ is a rescaling of $\gamma_2^*(\tilde{G})$ which is due to the fact that $\tilde{V}$ has range $\pm 1$.

The crucial idea in our alternative proof is to interpret the XOR game $G$ as an element of $\ell_{1}^{|\mathcal{X} |}(\ell_{\infty}^{|\mathcal{A}|}) \otimes \ell_{1}^{|\mathcal{Y} |}(\ell_{\infty}^{|\mathcal{B}|})$, with $|\mathcal{A}| = |\mathcal{B}| = 2$, instead of $\ell_{1}^{|\mathcal{X} |} \otimes \ell_{1}^{|\mathcal{Y} |}$.
Furthermore, if $G$ is an XOR game, Proposition \ref{lemmaUpperOmegaGamma} can be strengthened to
\begin{proposition}
\label{lemmaOmegaGammaEquiv1}
Let $G = (\pi,V)$ be an XOR game with $G \in \ell_{1}^{|\mathcal{X} |}(\ell_{\infty}^{|\mathcal{A}|}) \otimes \ell_{1}^{|\mathcal{Y} |}(\ell_{\infty}^{|\mathcal{B}|})$ and $|\mathcal{A}| = |\mathcal{B}| = 2$. Then
\[
\omega^*(G) = \gamma_2^*(G) ~.
\]
\end{proposition}
The proof can be found in Section \ref{sectionDualHilXor}.
We now have all the tools we need in order to give an alternative proof of the perfect parallel repetition theorem for entangled XOR games.
\begin{theorem}
\label{theoremXORGames}
Let $G = (\pi,V)$ be an XOR game with $G \in \ell_{1}^{|\mathcal{X} |}(\ell_{\infty}^{|\mathcal{A}|}) \otimes \ell_{1}^{|\mathcal{Y} |}(\ell_{\infty}^{|\mathcal{B}|})$ and $|\mathcal{A}| = |\mathcal{B}| = 2$. Then
\[
\omega^*(G^{\odot N}) = \omega^*(G)^N ~.
\]
\end{theorem}

\begin{proof}
It is clear that $\omega^*(G^{\odot N}) \geq \omega^*(G)^N$ by executing the rounds individually. For the other direction,
by Proposition \ref{lemmaUpperOmegaGamma}, we have
\[
\omega^*(G^{\odot N}) \leq \gamma_2^*(G^{\odot N}) ~.
\]
Applying the direct-product result of Theorem \ref{lemmaProductTheoremGenNorm1} and using Proposition \ref{lemmaOmegaGammaEquiv1} yields
\[
\omega^*(G^{\odot N}) \leq \gamma_2^*(G^{\odot N}) \leq  \gamma_2^*(G)^N = \omega^*(G)^N ~.
\]
\end{proof}

\subsection{Application II: Upper Bound on the Value of Two-Prover Games}
\label{sectionApp2}

In the following we will give an upper bound on the  maximal ratio between the entangled and the classical value of two-prover games, i.e., we will compute
\[
v  := \sup_G \left\{ \frac{\omega^*(G)}{\omega(G)} ~:~ G = (\pi, V) ~ \textit{two-prover game} \right\} ~.
\]
The best upper bound known so far \cite{commComplexNonSig} states that $v \leq O(|\mathcal{A}| \cdot |\mathcal{B}|)$ independently of the input dimensions $|\mathcal{X}|$ and $|\mathcal{Y}|$. If we fix the dimension of the local Hilbert spaces to $d$ in the computation of $\omega^*(G)$, it has been shown \cite{operatorSpace2} that $v \leq O(d)$, independently of the input and output dimensions. Note that these two results also hold if one considers the more general setting of Bell inequalities instead of two-prover games.
We prove a result which improves the previous upper bound of \cite{commComplexNonSig} by a square root factor.
\begin{theorem}
\label{theoremBellupperNew}
Let the input alphabet sizes $|\mathcal{X}|$, $|\mathcal{Y}|$ and the output alphabet sizes $|\mathcal{A}|$, $|\mathcal{B}|$ of the two-prover games be finite. Then
\[
v \leq  K \cdot \sqrt{|\mathcal{A}|  |\mathcal{B}|} ~,
\]
independently of the input dimensions $|\mathcal{X}|$ and $|\mathcal{Y}|$ and with $K = \frac{\pi}{2 \ln(1 + \sqrt{2})}$.
\end{theorem}

\begin{proof}
By using Proposition \ref{lemmaEqiuvClassVAr} and Proposition \ref{lemmaUpperOmegaGamma} and the dual of the generalized Grothendieck inequality in tensor form given in Theorem \ref{theoremDualGrothen} we get
\[
v = \sup_G \frac{\omega^*(G)}{\omega(G)} \leq \sup_G \frac{ \gamma_2^*(G)}{\varepsilon(G)} \leq  K \cdot \sqrt{|\mathcal{A}|  |\mathcal{B}|} ~,
\]
where the supremum is over two-prover games $G \in \ell_{1}^{|\mathcal{X} |}(\ell_{\infty}^{|\mathcal{A}|}) \otimes \ell_{1}^{|\mathcal{Y} |}(\ell_{\infty}^{|\mathcal{B}|})$.
\end{proof}

This theorem can be seen as a generalization of Tsirelson's work in \cite{tirelsonGrothendieckearly}. In particular, 
an XOR game $G = (\pi, V)$ can be interpreted as an element of $\ell_{1}^{|\mathcal{X}|} \otimes \ell_{1}^{|\mathcal{Y}|}$ with the respective correlation strategies $P \in \ell_{\infty}^{|\mathcal{X}|} \otimes \ell_{\infty}^{|\mathcal{Y}|}$ containing elements in $[-1,1]$ corresponding to expectation values. Proposition \ref{lemmaEqiuvClassVAr} and Proposition \ref{lemmaUpperOmegaGamma} (in this case even with equality) can also be proven for this setting of correlation strategies. Therefore, by using the standard Grothendieck inequality in dual tensor form (see also (\ref{eqnStandGrothTensor})) one obtains \cite{tirelsonGrothendieckearly}
\[
 \sup_{G} \frac{\omega^*(G)}{ \omega(G) } = \sup_{G} \frac{\gamma_2^*(G)}{ \varepsilon(G) } \leq K_G
\]
for $G$ an XOR game, independently of the input dimensions $|\mathcal{X}|$ and $|\mathcal{Y}|$.

\section{Concluding Remarks and Open Questions}

We have investigated the Hilbertian tensor norm $\gamma_2$ and its dual $\gamma_2^*$ defined over the tensor spaces $\ell_{\infty}^{|\mathcal{X} |}(\ell_{1}^{|\mathcal{A}|}) \otimes \ell_{\infty}^{|\mathcal{Y} |}(\ell_{1}^{|\mathcal{B}|})$ and $\ell_{1}^{|\mathcal{X} |}(\ell_{\infty}^{|\mathcal{A}|}) \otimes \ell_{1}^{|\mathcal{Y} |}(\ell_{\infty}^{|\mathcal{B}|})$.
We have given an alternative proof of the perfect parallel repetition theorem for entangled XOR games using our direct-product theorems for these tensor norms. Furthermore, by applying our generalized Grothendieck inequality we could establish an upper bound on the maximal ratio between entangled and classical values of two-prover games.

As a line of possible future work, it would be interesting to investigate if our generalization of Grothendieck's inequality and our direct-product theorems have other applications in quantum information theory, communication complexity or approximation algorithms.

\section*{Acknowledgements}

I thank Thomas Holenstein for helpful discussions and comments, Matthias Fitzi, Esther H\"anggi, Marco Tomamichel and Severin Winkler for useful comments on an earlier version of this paper, the referees for constructive comments that have significantly improved the presentation of this paper and Carlos Palazuelos for pointing out an error in a previous version of Theorem \ref{theoremBellupperNew}. This research was supported by the Swiss National Science Foundation through the National Centre of Competence in Research \emph{Quantum Science and Technology}.

\section{Proofs}

\subsection{Notation}
\label{sectionNotationProofs}

In addition to the notation introduced in Section \ref{sectionPreNot} we will use the following notation in this proof section.
Let $X$ and $Y$ be Banach spaces and $T:X \rightarrow Y$ be a linear map. The \emph{operator norm} of $T$ is defined as
\[
\| T \|_{X \rightarrow Y} := \sup_{v \in X} \{ \| T(v) \|_Y ~:~ \| v \|_X \leq 1 \} ~.
\]
Given column vectors $v_{i,j} \in \mathbb{R}^n$ with $1 \leq i \leq N$ and $1 \leq j \leq M$, we write $(v_{i,j})$ for the $n \times N \cdot M$-matrix which has the vectors $v_{i,j}$ as columns. The order is such that the first $M$ columns of $(v_{i,j})$ are given by the vectors $v_{1,1},v_{1,2},...,v_{1,M}$. The second $M$ columns are $v_{2,1},v_{2,2},...,v_{2,M}$, and so forth. We write $(v_{i,j})^T$ for the transposed matrix, i.e., the matrix $(v_{i,j})^T$ has the vectors $v_{i,j}^T$ as rows.

We will use $(\mu_{ij})$, with $1 \leq i \leq n$ and $1 \leq j \leq m$, to denote the $n \times m$-matrix with entries $\mu_{ij} \in \mathbb{R}$.

We write $sign : \mathbb{R} \rightarrow \{-1,+1\}$ for the sign-function, i.e., $sign(a) = +1$ if $a \geq 0$ and $sign(a) = -1$ if $a < 0$.

\subsection{Alternative Expressions for the $1(\infty) \rightarrow 2$ and $2 \rightarrow \infty(1)$ Operator Norms}

In this section we prove new equivalent expressions for the operator norms $\| \cdot \|_{1(\infty) \rightarrow 2}$ and $\| \cdot \|_{2 \rightarrow \infty(1)}$.
\begin{lemma}
\label{lemmaRepRSNormInf2}
Let $n_{y,b} \in \ell_2^n$ and $m_{x,a} \in \ell_2^n$, with $1 \leq n \leq \infty$, for all $1 \leq y \leq |\mathcal{Y}|$, $1 \leq b \leq |\mathcal{B}|$, $1 \leq x \leq |\mathcal{X}|$, and $1 \leq a \leq |\mathcal{A}|$. Then
\[
 \| (n_{y,b})^T \|_{2 \rightarrow \infty(1)} =  \max_{1 \leq y \leq | \mathcal{Y} |} \max_{ s } \left\{   \left\| \sum_{b=1}^{|\mathcal{B}|} s(b) \cdot n_{y,b} \right\|_2 ~:~ s: \{1,2,...,|\mathcal{B}|\} \rightarrow \{-1,+1\} \right\} ~,
\]
and
\[
 \| (m_{x,a}) \|_{1(\infty) \rightarrow 2} = \max_{1 \leq x \leq | \mathcal{X} |} \max_{ s }  \left\{    \left\| \sum_{a=1}^{|\mathcal{A}|} s(a) \cdot m_{x,a} \right\|_2 ~:~ s: \{1,2,...,|\mathcal{A}|\} \rightarrow \{-1,+1\} \right\} ~.
\]
In particular, $\| (n_{y,b})^T \|_{2 \rightarrow \infty(1)}^2 \geq \sum_{b=1}^{|\mathcal{B}|} \| n_{y,b} \|_2^2$ and $\| (m_{x,a}) \|_{1(\infty) \rightarrow 2}^2 \geq \sum_{a=1}^{|\mathcal{A}|} \| m_{x,a} \|_2^2$ for all $y \in \{1,2,...,|\mathcal{Y}| \}$ and all $x \in \{1,2,...,|\mathcal{X}| \}$, respectively.
\end{lemma}

\begin{proof}
By using the definition of the $\| \cdot \|_{1(\infty) \rightarrow 2}$-norm, we obtain
\[
\| (m_{x,a}) \|_{1(\infty) \rightarrow 2} = \sup_{ \| G \|_{1(\infty)} \leq 1} \left\| \sum_{x=1}^{|\mathcal{X}|} \sum_{a=1}^{|\mathcal{A}|} \langle G, e_x \otimes e_a \rangle \cdot m_{x,a} \right\|_2 ~.
\]
As every $G$ with $\| G \|_{1(\infty)} \leq 1$ can be written as $G_{x,a} \equiv \langle G, e_x \otimes e_a \rangle = \kappa_{x} \cdot \mu_{x,a}$, with $\sum_{x=1}^{|\mathcal{X}|} | \kappa_x | \leq 1$ and $\mu_{x,a} \in [-1,1]$, we get
\begin{eqnarray}
\label{eqnUpperM1inf2}
\| (m_{x,a}) \|_{1(\infty) \rightarrow 2}  &\leq& \sup_{ \| G \|_{1(\infty)} \leq 1}  \sum_{x=1}^{|\mathcal{X}|} | \kappa_x | \cdot \left\| \sum_{a=1}^{|\mathcal{A}|} \mu_{x,a} \cdot m_{x,a} \right\|_2 \nonumber \\
&\leq& \sup_{ \| G \|_{1(\infty)} \leq 1} \max_{1 \leq x \leq |\mathcal{X}|}  \left\| \sum_{a=1}^{|\mathcal{A}|} \mu_{x,a} \cdot m_{x,a} \right\|_2 \nonumber \\
&=& \sup_{ \| \mu \|_{\ell_{\infty}^{|\mathcal{A}|}} \leq 1} \max_{1 \leq x \leq |\mathcal{X}| }   \left\| \sum_{a=1}^{|\mathcal{A}|} \mu_a \cdot m_{x,a} \right\|_2 ~,
\end{eqnarray}
where we used the triangle inequality in the first line.
That $\| (m_{x,a}) \|_{1(\infty) \rightarrow 2}$ is greater or equal than the upper bound of  (\ref{eqnUpperM1inf2}) is obvious, by setting $\kappa_x = 1$ for the optimal $x$, and hence we have equality. That the optimal vector $\mu$ in  (\ref{eqnUpperM1inf2}) can be chosen to consist only of $+1, -1$ entries follows from the convexity of norms. That $\| (n_{y,b})^T \|_{2 \rightarrow \infty(1)} = \max_{ s, y }   \| \sum_{b=1}^{|\mathcal{B}|} s(b) \cdot n_{y,b} \|_2$ holds as well follows from Lemma \ref{lemmaDualityA} in Appendix \ref{sectionNotation}.

Let us now show that $\| (m_{x,a}) \|_{1(\infty) \rightarrow 2}^2 \geq \sum_{a=1}^{|\mathcal{A}|} \| m_{x,a} \|_2^2$.
By using the above result we obtain
\begin{eqnarray}
\| (m_{x,a}) \|_{1(\infty) \rightarrow 2}^2 &\geq&   \left\langle \sum_{a_1 = 1}^{|\mathcal{A}|}  s(a_1) \cdot m_{x,a_1}, \sum_{a_2=1}^{|\mathcal{A}|} s(a_2) \cdot m_{x,a_2} \right\rangle  \nonumber \\
&=& \sum_{a=1}^{|\mathcal{A}|} \langle m_{x,a}, m_{x,a} \rangle + \sum_{a_1 \neq a_2} s(a_1) \cdot s(a_2) \cdot \langle m_{x,a_1}, m_{x,a_2} \rangle ~, \nonumber
\end{eqnarray}
for all $x \in \{1,2,...,|\mathcal{X}| \}$ and all $s: \{1,2,...,|\mathcal{A}| \} \rightarrow \{-1, +1 \}$.
If we can show that there exists a function $s : \{1,2,...,|\mathcal{A}|\} \rightarrow \{-1,+1\}$ such that $\sum_{a_1 \neq a_2} s(a_1) \cdot s(a_2) \cdot \langle m_{x,a_1}, m_{x,a_2} \rangle \geq 0$, then we can conclude that 
\[
\| (m_{x,a}) \|_{1(\infty) \rightarrow 2}^2  \geq \sum_{a=1}^{|\mathcal{A}|} \langle m_{x,a}, m_{x,a} \rangle = \sum_{a=1}^{|\mathcal{A}|} \| m_{x,a} \|_2^2 ~.
\]
We will now construct a function with this property. First, we can write
\begin{equation}
\label{eqnRewrittenSumB}
 \sum_{a_1 \neq a_2} s(a_1) \cdot s(a_2) \cdot \langle m_{x,a_1}, m_{x,a_2} \rangle = 2 \cdot \sum_{a_1 = 2}^{|\mathcal{A}|} s(a_1) \cdot \left( \sum_{a_2 = 1}^{a_1 - 1} s(a_2) \cdot \langle m_{x,a_1}, m_{x,a_2} \rangle \right) ~.
\end{equation}
For $a = 1$ we set $s(1) := 1$. We then set the value for $s(2)$ which will depend on $s(1)$. Then we set $s(3)$ which will depend on $s(1)$ and $s(2)$. Hence, the value for $s(a_1)$ will depend on all $s(1),s(2),...,s(a_1 - 1)$. In particular, we define $s(a_1)$ to be
\[
s(a_1) := sign \left( \sum_{a_2 = 1}^{a_1 - 1} s(a_2) \cdot \langle m_{x,a_1}, m_{x,a_2} \rangle \right) ~.
\]
By defining the function $s$ in this way, the right hand side of  (\ref{eqnRewrittenSumB}) is always non-negative which is what we wanted to prove. By the same reasoning, one can show that $\| (n_{y,b})^T \|_{2 \rightarrow \infty(1)}^2 \geq \sum_{b=1}^{|\mathcal{B}|} \| n_{y,b} \|_2^2$ holds as well.
\end{proof}

Note that, for $| \mathcal{A} | = | \mathcal{B} | = 1$, we have that $\| R \|_{2 \rightarrow \infty}$ is the largest $2$-norm of a row of $R$ and $\| S \|_{1 \rightarrow 2}$ is the largest $2$-norm of a column of $S$.

\subsection{Generalized Grothendieck Inequality}
\label{sectionGeneralGrothen1}

In this section, we will state and prove a generalized Grothendieck inequality. In the tensor norm picture, it is a generalization of the standard Grothendieck inequality \cite{grothendieckInequality} in the sense that multiple outputs are allowed, or in the language of games, it generalizes from XOR games to arbitrary games. In the tensor norm language, the (generalized) Grothendieck inequality establishes a connection between the projective tensor norm $\pi$ and the Hilbertian tensor norm $\gamma_2$. 
The difference of our generalized Grothendieck inequality to the standard one is that the they are defined over different local Banach spaces.

By Lemma \ref{lemmaCrossnorm} in Appendix \ref{sectionProjInjNorm} we know that $\pi$ dominates $\gamma_2$, i.e., that $\pi(P) \geq \gamma_2(P)$ for all $P \in \ell_{\infty}^{|\mathcal{X} |}(\ell_{1}^{|\mathcal{A}|}) \otimes \ell_{\infty}^{|\mathcal{Y} |}(\ell_{1}^{|\mathcal{B}|})$. On the other hand, Grothendieck's inequality in tensor form upper bounds $\pi$ by $\gamma_2$, i.e., it is of the form
\[
\pi(P) \leq c \cdot \gamma_2(P)  ~,  ~ \forall P \in \ell_{\infty}^{|\mathcal{X} |}(\ell_{1}^{|\mathcal{A}|}) \otimes \ell_{\infty}^{|\mathcal{Y} |}(\ell_{1}^{|\mathcal{B}|}) ~,
\]
for $c = K_G$ and local Banach spaces $\ell_{\infty}^{|\mathcal{X} |}(\ell_{1}^{|\mathcal{A}|})$ and $\ell_{\infty}^{|\mathcal{Y} |}(\ell_{1}^{|\mathcal{B}|})$, with $|\mathcal{A}| = |\mathcal{B}| = 1$. Our goal in this section is to determine the best possible $c$ for local Banach spaces $\ell_{\infty}^{|\mathcal{X} |}(\ell_{1}^{|\mathcal{A}|})$ and $\ell_{\infty}^{|\mathcal{Y} |}(\ell_{1}^{|\mathcal{B}|})$ with \emph{arbitrary} output alphabet sizes $|\mathcal{A}|$ and $|\mathcal{B}|$.

Before we can prove the generalized Grothendieck inequality in tensor form we need an additional result:
\begin{lemma}[Alon \& Naor \cite{ApproxCutNorm}]
\label{lemmaAlonNaor}
For any sets $\{ x_i \}_{1 \leq i \leq n}$ and $\{ y_j \}_{1 \leq j \leq m}$ of real unit vectors in a Hilbert space $\mathcal{H}$, there are sets $\{ \tilde{x}_i \}_{1 \leq i \leq n}$ and $\{ \tilde{y}_j \}_{1 \leq j \leq m}$ of real unit vectors in a Hilbert space $\tilde{\mathcal{H}}$, such that
\[
\langle x_i , y_j \rangle = \frac{\pi}{2 \ln(1 + \sqrt{2})} \int_{\tilde{\mathcal{H}}} sign\langle \tilde{x}_i, z \rangle \cdot sign\langle \tilde{y}_j, z \rangle \gamma(dz) ~,  
\]
for all $1 \leq i \leq n$, $1 \leq j \leq m$, where $\gamma(dz)$ is the normalized Gauss measure on $\tilde{\mathcal{H}}$.
\end{lemma}

\begin{claim}[Generalized Grothendieck Inequality in Tensor Form]
\label{theoremGenerGrothenTensor}
For any $P \in \ell_{\infty}^{|\mathcal{X} |}(\ell_{1}^{|\mathcal{A}|}) \otimes \ell_{\infty}^{|\mathcal{Y} |}(\ell_{1}^{|\mathcal{B}|})$ it holds that
\[
\pi(P) \leq K \cdot \sqrt{|\mathcal{A}| |\mathcal{B}|} \cdot \gamma_2(P) ~,
\]
with $K = \frac{\pi}{2 \ln(1 + \sqrt{2})}$. 
\end{claim}

\begin{proof}
Let us assume that $\gamma_2(P) = 1$ for some $P = \sum_{x,y,a,b} P_{x,y}^{a,b} \cdot (e_x \otimes e_a) \otimes (e_y \otimes e_b)$. Showing that $\pi(P) \leq \frac{\pi}{2 \ln(1 + \sqrt{2})} \cdot \sqrt{|\mathcal{A}| |\mathcal{B}|}$ proves the claim.

As $\gamma_2(P) = 1$, we can conclude according to Corollary \ref{corollaryEqivGammaVec} in Appendix \ref{sectionHilbertTensor1}, that there exist real vectors $\{ m_{x,a} \}$ and $\{ n_{y,b} \}$ in $\ell_2$ with
\[
\| (m_{x,a}) \|_{1(\infty) \rightarrow 2} \leq 1 ~,~ \| (n_{y,b})^T \|_{2 \rightarrow \infty(1)} \leq 1 ~,
\]
such that
\begin{equation}
\label{eqnRepRhoTensmn}
P = \sum_{x,y,a,b} \langle m_{x,a} , n_{y,b} \rangle \cdot e_{x,a} \otimes e_{y,b} ~, 
\end{equation}
with $e_{x,a} := e_x \otimes e_a$ and $e_{y,b} := e_y \otimes e_b$. Applying the second part of Lemma \ref{lemmaRepRSNormInf2} yields
\[
\sum_{a=1}^{|\mathcal{A}|} \| m_{x,a} \|_2^2 \leq 1 ~,~ \sum_{b=1}^{|\mathcal{B}|} \| n_{y,b} \|_2^2 \leq 1 ~,
\]
for all $x \in \{1,2,...,|\mathcal{X}| \}$ and $y \in \{1,2,...,|\mathcal{Y}| \}$.

Using Lemma \ref{lemmaAlonNaor} on the vectors $\{ m_{x,a}\}$ and $\{ n_{y,b} \}$ implies
\begin{equation}
\label{eqnInnerRepAlon}
 \left\langle \frac{m_{x,a}}{\| m_{x,a} \|_2}, \frac{n_{y,b}}{\| n_{y,b} \|_2} \right\rangle = \frac{\pi}{2 \ln(1 + \sqrt{2})} \int_{\tilde{\mathcal{H}}} sign\langle \tilde{m}_{x,a}, z \rangle \cdot sign\langle \tilde{n}_{y,b}, z \rangle \gamma(dz) ~,
\end{equation}
for all $x \in \{1,2,...,|\mathcal{X}|\}, y \in \{1,2,...,|\mathcal{Y}|\}, a \in \{1,2,...,|\mathcal{A}|\}$ and $b \in \{1,2,...,|\mathcal{B}|\}$. Combining  (\ref{eqnRepRhoTensmn}) and (\ref{eqnInnerRepAlon}) gives
\[
 P = c \int_{\tilde{\mathcal{H}}} \left( \sum_{x,a} \| m_{x,a} \|_2 \cdot  sign\langle \tilde{m}_{x,a}, z \rangle  e_{x,a} \right) \otimes \left( \sum_{y,b} \| n_{y,b} \|_2 \cdot   sign\langle \tilde{n}_{y,b}, z \rangle  e_{y,b} \right) \gamma(dz) ~,
\]
with $c = \frac{\pi}{2 \ln(1 + \sqrt{2})}$. Since $\pi$ is a norm, we can apply the triangle inequality and get
\[
 \pi(P) \leq c \cdot \sup_{z} \left\| \sum_{x,a} \| m_{x,a} \|_2 \cdot  sign\langle \tilde{m}_{x,a}, z \rangle  e_{x,a} \right\|_{\infty(1)}   \left\| \sum_{y,b} \| n_{y,b} \|_2 \cdot   sign\langle \tilde{n}_{y,b}, z \rangle  e_{y,b} \right\|_{\infty(1)} ~,
\]
where we also used that $\pi$ is a tensor norm and therefore $\pi(P_A \otimes P_B) = \| P_A \|_{\infty(1)} \cdot \| P_B \|_{\infty(1)}$.
Furthermore, by using the definition of the $\infty(1)$-norm, we have that
\[
\left\| \sum_{x,a} \| m_{x,a} \|_2 \cdot  sign\langle \tilde{m}_{x,a}, z \rangle  e_{x,a} \right\|_{\infty(1)} \leq \max_{x \in \{1,2,...,|\mathcal{X}| \}} \sum_{a=1}^{|\mathcal{A}|} \| m_{x,a} \|_2 ~,
\]
for any $z$. Using that $\sum_{a=1}^{|\mathcal{A}|} \| m_{x,a} \|_2^2 \leq 1$ implies $\sum_{a=1}^{|\mathcal{A}|} \| m_{x,a} \|_2 \leq \sqrt{|\mathcal{A}|}$ (by Cauchy-Schwarz inequality) finishes the proof.
\end{proof}

\setcounter{theorem}{0}
By applying Lemma \ref{lemmaTakingDuals} in Appendix \ref{sectionNotation} we get the following dual theorem:
\begin{theorem}[Generalized Grothendieck Inequality in Dual Tensor Form]
\label{corollaryDualGrothen}
For any $G \in \ell_{1}^{|\mathcal{X} |}(\ell_{\infty}^{|\mathcal{A}|}) \otimes \ell_{1}^{|\mathcal{Y} |}(\ell_{\infty}^{|\mathcal{B}|})$ it holds that
\[
\gamma_2^*(G) \leq K \cdot \sqrt{|\mathcal{A}| |\mathcal{B}|} \cdot \varepsilon(G) ~,
\]
with $K = \frac{\pi}{2 \ln(1 + \sqrt{2})}$.
\end{theorem}

As the standard Grothendieck inequality is usually stated in matrix form, we will also give a matrix representation of our generalization.
As the $\| \cdot \|_{2 \rightarrow \infty(1)}$ and $\| \cdot \|_{1(\infty) \rightarrow 2}$ operator norms will appear in the following claim, it might be helpful for the reader to have a look at Lemma \ref{lemmaRepRSNormInf2} again which gives an alternative representation of these two operator norms.
\begin{claim}[Generalized Grothendieck Inequality in Matrix Form]
\label{theoremGeneralGrothenMatrix}
For any set of real numbers $\{ \alpha_{x,y}^{a,b} \}$, with $1 \leq x \leq |\mathcal{X}|$, $1 \leq y \leq |\mathcal{Y}|$, $1 \leq a \leq |\mathcal{A}|$, and $1 \leq b \leq |\mathcal{B}|$, it holds that
\begin{eqnarray}
 &\sup& \left\{ \left| \sum_{x,y,a,b} \alpha_{x,y}^{a,b} \cdot \langle m_{x,a}, n_{y,b} \rangle \right| ~:~ \| (n_{y,b})^T \|_{2 \rightarrow \infty(1)} \leq 1 ~,~ \| (m_{x,a}) \|_{1(\infty) \rightarrow 2} \leq 1 \right\} \nonumber \\
 &\leq& K \cdot \sqrt{|\mathcal{A}| |\mathcal{B}|} \cdot \sup \left\{ \left| \sum_{x,y,a,b} \alpha_{x,y}^{a,b} \cdot s_{x,a} \cdot t_{y,b} \right| ~:~ \| (s_{x,a}) \|_{\infty(1)} \leq 1 ~,~ \| (t_{y,b}) \|_{\infty(1)} \leq 1 \right\} \nonumber ~,
\end{eqnarray}
with $K = \frac{\pi}{2 \ln(1 + \sqrt{2})}$ and where the supremum is over real vectors $m_{x,a}, n_{y,b} \in \ell_2^n$, with $1 \leq n \leq \infty$, and real numbers $s_{x,a}, t_{y,b} \in [-1,+1]$ with $\| (s_{x,a}) \|_{\infty(1)} = \max_{1 \leq x \leq |\mathcal{X}| } \sum_{a=1}^{|\mathcal{A}|} | s_{x,a} |$ and $\| (t_{y,b}) \|_{\infty(1)} = \max_{1 \leq y \leq |\mathcal{Y}| } \sum_{b=1}^{|\mathcal{B}|} | t_{y,b} |$.
\end{claim}

\begin{proof}
Let $G \in \ell_{1}^{|\mathcal{X} |}(\ell_{\infty}^{|\mathcal{A}|}) \otimes \ell_{1}^{|\mathcal{Y} |}(\ell_{\infty}^{|\mathcal{B}|})$ with 
\[
\langle G, e_{x,a} \otimes e_{y,b} \rangle := \alpha_{x,y}^{a,b} ~,
\]
and $P_s = \sum_{x,a} s_{x,a} \cdot e_{x,a} \in \ell_{\infty}^{|\mathcal{X} |}(\ell_{1}^{|\mathcal{A}|})$ and $P_t = \sum_{y,b} t_{y,b} \cdot e_{y,b} \in \ell_{\infty}^{|\mathcal{Y} |}(\ell_{1}^{|\mathcal{B}|})$ with $s_{x,a}, t_{y,b} \in \mathbb{R}$.
Computing the injective tensor norm of $G$ yields
\begin{eqnarray}
\label{eqnVarMRepOth}
 \varepsilon(G) &=& \sup \{ | \langle G, P_s \otimes P_t \rangle | ~:~ \| P_s \|_{\infty(1)} \leq 1~,~ \| P_t \|_{\infty(1)} \leq 1 \} \nonumber \\
 &=& \sup \left\{ \left| \sum_{x,y,a,b} \langle G, e_{x,a} \otimes e_{y,b} \rangle \cdot s_{x,a} \cdot t_{y,b} \right| ~:~ \| (s_{x,a}) \|_{\infty(1)} \leq 1~,~ \| (t_{y,b}) \|_{\infty(1)} \leq 1 \right\} \nonumber \\
 &=& \sup \left\{ \left| \sum_{x,y,a,b} \alpha_{x,y}^{a,b} \cdot s_{x,a} \cdot t_{y,b} \right| ~:~ \| (s_{x,a}) \|_{\infty(1)} \leq 1~,~ \| (t_{y,b}) \|_{\infty(1)} \leq 1 \right\} ~.
\end{eqnarray}
On the other hand, using Corollary \ref{corollaryEqivGammaVec} in Appendix \ref{sectionHilbertTensor1}, we obtain
\begin{eqnarray}
\label{eqnGammRepDif}
 \gamma_2^*(G) &=& \sup \{ | \langle G, P \rangle | ~:~ \gamma_2(P) \leq 1 \} \nonumber \\
 &=& \sup \{ | \langle G, P \rangle | ~:~ \| (m_{x,y}) \|_{1(\infty) \rightarrow 2} \leq 1~,~ \| (n_{y,b})^T \|_{2 \rightarrow \infty(1)} \leq 1 \} \nonumber \\
 &=& \sup \left\{ \left| \sum_{x,y,a,b} \alpha_{x,y}^{a,b}  \langle m_{x,a}, n_{y,b} \rangle \right| ~:~ \| (m_{x,y}) \|_{1(\infty) \rightarrow 2} \leq 1 ,~\| (n_{y,b})^T \|_{2 \rightarrow \infty(1)} \leq 1 \right\} ~,
\end{eqnarray}
with $P = \sum_{x,y,a,b} P_{x,y}^{a,b} \cdot e_{x,a} \otimes e_{y,b} \in \ell_{\infty}^{|\mathcal{X} |}(\ell_{1}^{|\mathcal{A}|}) \otimes \ell_{\infty}^{|\mathcal{Y} |}(\ell_{1}^{|\mathcal{B}|})$ and $P_{x,y}^{a,b} := \langle m_{x,a}, n_{y,b} \rangle$. Equations (\ref{eqnVarMRepOth}) and (\ref{eqnGammRepDif}) together with Theorem \ref{corollaryDualGrothen} yield the result.
\end{proof}

The standard Grothendieck inequality in tensor as well as in matrix form are recovered from our generalized Grothendieck inequalities by setting $| \mathcal{A}| = |\mathcal{B}| = 1$.

\subsection{Direct-Product Theorems}
\label{sectionDirctProductTheo}

We first show a direct-product result for the $\gamma_2^*$ tensor norm over Banach spaces which have the property that their norms behave ''nicely`` on product tensors. By behaving ''nicely`` we mean the following.
Let $X_n := (\mathbb{R}^n, \| \cdot \|_{X_n})$, with $1 \leq n < \infty$, be a Banach space. Then the norm $\| \cdot \|_{X_n}$ behaves ''nicely`` on product tensors if
\[
\| P_A \otimes P_B \|_{X_{n\cdot m}}  \leq \| P_A \|_{X_n} \cdot \| P_B \|_{X_m}  ~,
\]
for all $1 \leq n, m < \infty$, $P_A \in \mathbb{R}^n$, $P_B \in \mathbb{R}^m$ and with $P_A \otimes P_B \in \mathbb{R}^{n \cdot m}$. To shorten the notation we will usually write
\[
\| P_A \otimes P_B \|_{X}  \leq \| P_A \|_{X} \cdot \| P_B \|_{X}  ~.
\]
Let us now show that the $\infty(1)$-norm and its dual have the property that they behave ''nicely`` on product tensors.
First, the $\infty(1)$-norm can also be defined for tensor elements $P_{A_1} \otimes P_{A_2} \in \ell_{\infty}^{|\mathcal{X}_1|}(\ell_{1}^{|\mathcal{A}_1|}) \otimes \ell_{\infty}^{|\mathcal{X}_2|}(\ell_{1}^{|\mathcal{A}_2|})$ by
\begin{equation}
\label{eqnDefInf1NormTensor}
\| P_{A_1} \otimes P_{A_2} \|_{\infty(1)} := \max_{x_1 \in \mathcal{X}_1, x_2 \in \mathcal{X}_2} \sum_{a_1 \in \mathcal{A}_1, a_2 \in \mathcal{A}_2} | \langle f_{x_1,a_1} \otimes f_{x_2,a_2},  P_{A_1} \otimes P_{A_2} \rangle |  ~,
\end{equation}
with $P_{A_1} \in \ell_{\infty}^{|\mathcal{X}_1|}(\ell_{1}^{|\mathcal{A}_1|})$ and $P_{A_2} \in \ell_{\infty}^{|\mathcal{X}_2|}(\ell_{1}^{|\mathcal{A}_2|})$. Then, we obtain:
\begin{lemma}
\label{lemmaInf1NormProduct}
Let $P_{A_1} \in \ell_{\infty}^{| \mathcal{X}_1 |}(\ell_{1}^{| \mathcal{A}_1 |})$, $P_{A_2} \in \ell_{\infty}^{| \mathcal{X}_2 |}(\ell_{1}^{| \mathcal{A}_2 |})$ and $G_{A_1} \in \ell_1^{| \mathcal{X}_1 |}(\ell_{\infty}^{| \mathcal{A}_1 |})$, $G_{A_2} \in \ell_1^{| \mathcal{X}_2 |}(\ell_{\infty}^{| \mathcal{A}_2 |})$. Then
\begin{eqnarray}
\| P_{A_1} \otimes P_{A_2} \|_{\infty(1)} &=& \| P_{A_1} \|_{\infty(1)} \cdot \| P_{A_2} \|_{\infty(1)} ~, \nonumber \\
\| G_{A_1} \otimes G_{A_2} \|_{1(\infty)} &=& \| G_{A_1} \|_{1(\infty)} \cdot \| G_{A_2} \|_{1(\infty)} ~, \nonumber
\end{eqnarray}
with $P_{A_1} \otimes P_{A_2} \in \ell_{\infty}^{|\mathcal{X}_1|}(\ell_{1}^{|\mathcal{A}_1|}) \otimes \ell_{\infty}^{|\mathcal{X}_2|}(\ell_{1}^{|\mathcal{A}_2|})$ and $G_{A_1} \otimes G_{A_2} \in \ell_{1}^{|\mathcal{X}_1|}(\ell_{\infty}^{|\mathcal{A}_1|}) \otimes \ell_{1}^{|\mathcal{X}_2|}(\ell_{\infty}^{|\mathcal{A}_2|})$.
\end{lemma}

\begin{proof}
Using the definition of the $\infty(1)$-norm (see  (\ref{eqnDefInf1NormTensor})) gives us
\begin{eqnarray}
\| P_{A_1} \otimes P_{A_2} \|_{\infty(1)} &=& \sup_{x_1,x_2} \sum_{a_1,a_2} | \langle f_{x_1,a_1} \otimes f_{x_2,a_2}, P_{A_1} \otimes P_{A_2} \rangle | \nonumber \\
&=& \sup_{x_1,x_2} \sum_{a_1,a_2} | \langle f_{x_1,a_1}, P_{A_1} \rangle | \cdot | \langle f_{x_2,a_2}, P_{A_2} \rangle | \nonumber \\
&=& \sup_{x_1} \sum_{a_1} | \langle f_{x_1,a_1}, P_{A_1} \rangle | \cdot \sup_{x_2} \sum_{a_2} | \langle f_{x_2,a_2}, P_{A_2} \rangle | \nonumber \\
&=& \| P_{A_1} \|_{\infty(1)} \cdot \| P_{A_2} \|_{\infty(1)} ~. \nonumber
\end{eqnarray}
Similarly, we get $\| G_{A_1} \otimes G_{A_2} \|_{1(\infty)} = \| G_{A_1} \|_{1(\infty)} \cdot \| G_{A_2} \|_{1(\infty)}$.
\end{proof}

Let us now prove a direct product result for the $\gamma_2^*$ tensor norm. 
We will need the following result in order to show this result.
\begin{lemma}[Bennett \cite{schurMultipliers}]
\label{lemmaProductNorm}
Let $A$ and $B$ be $n \times n$ and $m \times m$ matrices over $\mathbb{R}$, respectively. Then
\[
\| A \otimes B \|_{2 \rightarrow 2} = \| A \|_{2 \rightarrow 2} \cdot \| B \|_{2 \rightarrow 2} ~.
\]
\end{lemma}

\begin{lemma}
\label{lemmaGammaComposition}
Let $X$ and $Y$ be Banach spaces with the norms having the property that
\begin{eqnarray}
\| G_{A_1} \otimes G_{A_2} \|_X  &\leq& \| G_{A_1} \|_X \cdot \| G_{A_2} \|_X ~,  \nonumber \\
\| G_{B_1} \otimes G_{B_2} \|_Y  &\leq& \| G_{B_1} \|_Y \cdot \| G_{B_2} \|_Y ~, \nonumber
\end{eqnarray}
and let the tensor norm $\gamma_2^*$ be defined over the tensor space $X \otimes Y$.
Then
\[
\gamma_2^*(G_{A_1B_1} \odot G_{A_2B_2}) \leq  \gamma_2^*(G_{A_1B_1}) \cdot \gamma_2^*(G_{A_2B_2}) ~,
\]
where the partition of $G_{A_1B_1} \odot G_{A_2B_2}$ is with respect to $A_1A_2 : B_1B_2$.
\end{lemma}

\begin{proof}
Let $G_{A_1B_1} = \sum_{i,j} \alpha_{ij} \cdot G_{A_1}^i \otimes G_{B_1}^j$ and $G_{A_2B_2} = \sum_{k,l} \beta_{kl} \cdot G_{A_2}^k \otimes G_{B_2}^l$ be optimal decompositions in the definition of the $\gamma_2^*$ tensor norm (see  (\ref{eqnDefGammaDual}) in Appendix \ref{sectionHilbertTensor1}). Taking the composition of these two systems gives us (see also (\ref{eqnDefDecomp}))
\[
G_{A_1B_1} \odot G_{A_2B_2} = \sum_{i,k} \sum_{j,l} \alpha_{ij} \cdot \beta_{kl} \cdot (G_{A_1}^i \otimes G_{A_2}^k) \otimes (G_{B_1}^j \otimes G_{B_2}^l) ~.
\]
We therefore get
\begin{eqnarray}
 \gamma_2^*(G_{A_1B_1} \odot G_{A_2B_2}) &\leq& \| (\alpha_{ij} \cdot \beta_{kl}) \|_{2 \rightarrow 2} \cdot \ell_2(G_{A_1}^i \otimes G_{A_2}^k; X) \cdot \ell_2( G_{B_1}^j \otimes G_{B_2}^l ; Y) \nonumber \\
&=& \| (\alpha_{ij}) \otimes (\beta_{kl}) \|_{2 \rightarrow 2} \cdot \ell_2(G_{A_1}^i \otimes G_{A_2}^k; X) \cdot \ell_2( G_{B_1}^j \otimes G_{B_2}^l ; Y) \nonumber \\
&=& \| (\alpha_{ij}) \|_{2 \rightarrow 2} \cdot \| (\beta_{kl}) \|_{2 \rightarrow 2} \cdot \ell_2(G_{A_1}^i \otimes G_{A_2}^k; X) \cdot \ell_2( G_{B_1}^j \otimes G_{B_2}^l ; Y) ~, \nonumber
\end{eqnarray}
where we used Lemma \ref{lemmaProductNorm} in the last line. The fact that the local norms behave nicely on product tensors implies immediately that $\ell_2(G_{A_1}^i \otimes G_{A_2}^k; X) \leq \ell_2(G_{A_1}^i; X) \cdot \ell_2(G_{A_2}^k; X)$
and $\ell_2(G_{B_1}^j \otimes G_{B_2}^l; Y) \leq \ell_2(G_{B_1}^j; Y) \cdot \ell_2(G_{B_2}^l; Y)$, and therefore
\begin{eqnarray}
\gamma_2^*(G_{A_1B_1} \odot G_{A_2B_2}) &\leq&  \| (\alpha_{ij}) \|_{2 \rightarrow 2} \cdot \ell_2(G_{A_1}^i; X) \cdot \ell_2(G_{B_1}^j; Y) \nonumber \\
&\cdot& \| (\beta_{kl}) \|_{2 \rightarrow 2} \cdot \ell_2(G_{A_2}^k; X) \cdot \ell_2(G_{B_2}^l; Y) \nonumber \\
&=& \gamma_2^*(G_{A_1B_1}) \cdot \gamma_2^*(G_{A_2B_2}) ~. \nonumber
\end{eqnarray}
\end{proof}

The next theorem gives new direct-product results for the $\gamma_2$ and $\gamma_2^*$ tensor norms over the Banach spaces $\ell_{\infty}^{|\mathcal{X} |}(\ell_{1}^{|\mathcal{A}|}) \otimes \ell_{\infty}^{|\mathcal{Y} |}(\ell_{1}^{|\mathcal{B}|})$ and $\ell_{1}^{|\mathcal{X} |}(\ell_{\infty}^{|\mathcal{A}|}) \otimes \ell_{1}^{|\mathcal{Y} |}(\ell_{\infty}^{|\mathcal{B}|})$.
\begin{theorem}
\label{lemmaProductTheoremGenNorm}
Let $P_{A_1B_1} \in \ell_{\infty}^{|\mathcal{X}_1 |}(\ell_{1}^{|\mathcal{A}_1|}) \otimes \ell_{\infty}^{|\mathcal{Y}_1 |}(\ell_{1}^{|\mathcal{B}_1|})$, $P_{A_2B_2} \in \ell_{\infty}^{|\mathcal{X}_2 |}(\ell_{1}^{|\mathcal{A}_2|}) \otimes \ell_{\infty}^{|\mathcal{Y}_2 |}(\ell_{1}^{|\mathcal{B}_2|})$, $G_{A_1B_1} \in \ell_{1}^{|\mathcal{X}_1 |}(\ell_{\infty}^{|\mathcal{A}_1|}) \otimes \ell_{1}^{|\mathcal{Y}_1 |}(\ell_{\infty}^{|\mathcal{B}_1|})$, $G_{A_2B_2} \in \ell_{1}^{|\mathcal{X}_2 |}(\ell_{\infty}^{|\mathcal{A}_2|}) \otimes \ell_{1}^{|\mathcal{Y}_2 |}(\ell_{\infty}^{|\mathcal{B}_2|})$ and $P_{A_1B_1} \odot P_{A_2B_2}$ and $G_{A_1B_1} \odot G_{A_2B_2}$ be bipartite with respect to the partition $A_1A_2:B_1B_2$. Then
\begin{eqnarray}
\gamma_2(P_{A_1B_1} \odot P_{A_2B_2}) &\geq& \gamma_2(P_{A_1B_1}) \cdot \gamma_2(P_{A_2B_2}) ~, \nonumber \\
\gamma_2(G_{A_1B_1} \odot G_{A_2B_2}) &\geq& \gamma_2(G_{A_1B_1}) \cdot \gamma_2(G_{A_2B_2}) ~, \nonumber \\
\gamma_2^*(G_{A_1B_1} \odot G_{A_2B_2}) &\leq& \gamma_2^*(G_{A_1B_1}) \cdot \gamma_2^*(G_{A_2B_2}) ~, \nonumber \\
\gamma_2^*(P_{A_1B_1} \odot P_{A_2B_2}) &\leq& \gamma_2^*(P_{A_1B_1}) \cdot \gamma_2^*(P_{A_2B_2}) \nonumber ~.
\end{eqnarray}
\end{theorem}

\begin{proof}
As $\| P_A \otimes P_B \|_{\infty(1)} = \| P_A \|_{\infty(1)} \cdot \| P_B \|_{\infty(1)}$ and $\| G_A \otimes G_B \|_{1(\infty)} = \| G_A \|_{1(\infty)} \cdot \| G_B \|_{1(\infty)}$ (by Lemma \ref{lemmaInf1NormProduct}), Lemma \ref{lemmaGammaComposition} immediately implies that $\gamma_2^*(G_{A_1B_1} \odot G_{A_2B_2}) \leq \gamma_2^*(G_{A_1B_1}) \cdot \gamma_2^*(G_{A_2B_2})$ and $\gamma_2^*(P_{A_1B_1} \odot P_{A_2B_2}) \leq \gamma_2^*(P_{A_1B_1}) \cdot \gamma_2^*(P_{A_2B_2})$ which, by duality, imply $\gamma_2(P_{A_1B_1} \odot P_{A_2B_2}) \geq \gamma_2(P_{A_1B_1}) \cdot \gamma_2(P_{A_2B_2})$ and $\gamma_2(G_{A_1B_1} \odot G_{A_2B_2}) \geq \gamma_2(G_{A_1B_1}) \cdot \gamma_2(G_{A_2B_2})$, respectively. This holds since
\begin{eqnarray}
 \gamma_2(P_{A_1B_1} \odot P_{A_2B_2}) &=& \sup \{ | \langle G, P_{A_1B_1} \odot P_{A_2B_2} \rangle | ~:~ \gamma_2^*(G) \leq 1 \} \nonumber \\
&\geq& \sup \{ | \langle G_1 \odot G_2, P_{A_1B_1} \odot P_{A_2B_2} \rangle | ~:~ \gamma_2^*(G_1) \leq 1~,~ \gamma_2^*(G_2) \leq 1 \} \nonumber \\
&=& \sup \{ | \langle G_1,P_{A_1B_1} \rangle | \cdot | \langle G_2, P_{A_2B_2} \rangle | ~:~ \gamma_2^*(G_1) \leq 1~,~ \gamma_2^*(G_2) \leq 1 \} \nonumber \\
&=& \gamma_2(P_{A_1B_1}) \cdot \gamma_2(P_{A_2B_2}) ~, \nonumber
\end{eqnarray}
where, in the second line, we used that $\gamma_2^*(G_1) \leq 1$ and $\gamma_2^*(G_2) \leq 1$ imply $\gamma_2^*(G_1 \odot G_2) \leq 1$.
\end{proof}

Theorem \ref{lemmaProductTheoremGenNorm} can be strengthened when restricting to the case where $|\mathcal{A}_1 | = |\mathcal{A}_2 |  =  |\mathcal{B}_1 | = |\mathcal{B}_2 |  = 1$, namely, one can obtain perfect direct-product theorems, i.e., equalities instead of upper or lower bounds \cite{paralleXOR, directProductTheorems}.

\subsection{Hilbertian Tensor Norm and Bipartite Quantum Systems}
\label{sectionHilbertQuantum}

Remember that we call $P \in \ell_{\infty}^{|\mathcal{X} |}(\ell_{1}^{|\mathcal{A}|}) \otimes \ell_{\infty}^{|\mathcal{Y} |}(\ell_{1}^{|\mathcal{B}|})$ a \emph{quantum system} if it can be obtained by measurements on a pure quantum state (see Section \ref{sectionHilbertTensor}).
On the other hand, we will call $P \in \ell_{\infty}^{|\mathcal{X} |} \otimes \ell_{\infty}^{|\mathcal{Y} |}$ a \emph{quantum correlation} if there exists a pure quantum state $| \Psi \rangle  \in \mathcal{H}_A \otimes \mathcal{H}_B$, and observables $A_1, ...,A_{|\mathcal{X}|}$ and $B_1, ...,B_{|\mathcal{Y}|}$ on $\mathcal{H}_A$ and $\mathcal{H}_B$, respectively, with eigenvalues $\pm 1$ such that
\[
\langle f_{x} \otimes f_{y}, P \rangle = \langle \Psi | A_x \otimes B_y | \Psi \rangle ~.
\]

Let us first focus on the case where $P$ is a quantum correlation. In order to establish a connection to the $\gamma_2$ tensor norm, we need a theorem by Tsirelson, which says that the correlations which can be obtained by measurements on a quantum state can be represented by inner products of real unit vectors, and vice versa. More formally:
\begin{lemma}[Tsirelson's Theorem \cite{tsirelson}]
\label{theoremTsirelson}
Let $A_1,...,A_{|\mathcal{X}|}$ and $B_1,...,B_{|\mathcal{Y}|}$ be observables with eigenvalues in $[ -1 , +1 ]$. Then for any state $| \Psi \rangle  \in \mathcal{H}_A \otimes \mathcal{H}_B$ there exist real unit vectors $m_1,...,m_{|\mathcal{X}|} \in \mathbb{R}^{2 \cdot \max\{|\mathcal{X}|, |\mathcal{Y}|\}}$ and $n_1,...,n_{|\mathcal{Y}|} \in \mathbb{R}^{2 \cdot \max\{|\mathcal{X}|, |\mathcal{Y}|\}}$ such that
\[
\langle m_x, n_y \rangle = \langle \Psi | A_x \otimes B_y | \Psi \rangle ~,
\]
for all $1 \leq x \leq |\mathcal{X}|$ and  $1 \leq y \leq |\mathcal{Y}|$.

Conversely, let $m_1,..,m_{|\mathcal{X}|}, n_1,...,n_{|\mathcal{Y}|} \in \mathbb{R}^N$ be real vectors with $\| m_x \|_2 \leq 1$ and $\| n_y \|_2 \leq 1$ for all $x \in \mathcal{X}$ and $y \in \mathcal{Y}$, respectively, and $| \Psi \rangle  \in \mathcal{H}_A \otimes \mathcal{H}_B$ be any maximally entangled state where $\dim(\mathcal{H}_A) =  \dim(\mathcal{H}_B) = 2^{\lceil N/2 + 1 \rceil}$. Then,  there exist observables $A_1,...,A_{|\mathcal{X}|}$ on $\mathcal{H}_A$ and $B_1,...,B_{|\mathcal{Y}|}$ on $\mathcal{H}_B$ with eigenvalues $\pm 1$ such that
\[
\langle m_x, n_y \rangle = \langle \Psi | A_x \otimes B_y | \Psi \rangle ~,
\]
for all $1 \leq x \leq |\mathcal{X}|$ and  $1 \leq y \leq |\mathcal{Y}|$.
\end{lemma}
Note that this is a slightly generalized version of Tsirelson's theorem where we do not need the vectors $m_x \in \mathbb{R}^N$ and $n_y \in \mathbb{R}^N$ to be unit vectors. So let us show that Lemma \ref{theoremTsirelson} indeed holds. In order to be allowed to apply the standard Tsirelson theorem, we need unit vectors. So let us construct them. Define $\tilde{m}_x \in \mathbb{R}^{N+2}$ to be $\langle \tilde{m}_x, e_i \rangle := \langle m_x, e_i \rangle$ for all $1 \leq i \leq N$, $\langle \tilde{m}_x, e_{N+1} \rangle := \sqrt{ 1 - \| m_x \|_2^2}$ and $\langle \tilde{m}_x, e_{N+2} \rangle := 0$. And similarly, for $\tilde{n}_y \in \mathbb{R}^{N+2}$ we set $\langle \tilde{n}_y, e_i \rangle := \langle n_y, e_i \rangle$ for all $1 \leq i \leq N$, $\langle \tilde{n}_y, e_{N+1} \rangle := 0$ and $\langle \tilde{n}_y, e_{N+2} \rangle := \sqrt{ 1 - \| n_y \|_2^2}$. We then have $\langle m_x, n_y \rangle = \langle \tilde{m}_x, \tilde{n}_y \rangle$ and $\| \tilde{m}_x \|_2 = 1$ and $\| \tilde{n}_y \|_2 = 1$ for all $1 \leq x \leq |\mathcal{X}|$ and  $1 \leq y \leq |\mathcal{Y}|$. Hence, we can apply the standard Tsirelson theorem and get $\langle m_x, n_y \rangle = \langle \tilde{m}_x, \tilde{n}_y \rangle = \langle \Psi | A_x \otimes B_y | \Psi \rangle$.

Using Lemma \ref{theoremTsirelson}, we are now ready to prove a tight connection between the $\gamma_2$ tensor norm and quantum correlations, i.e.,
\begin{lemma}
\label{lemmaQuantumCorrGamma}
$P \in \ell_{\infty}^{|\mathcal{X} |} \otimes \ell_{\infty}^{|\mathcal{Y} |}$ is a quantum correlation if and only if $\gamma_2(P) \leq 1$.
\end{lemma}

\begin{proof}
As $P$ is a quantum correlation we can write it, according to Lemma \ref{theoremTsirelson}, as 
\[
\langle f_x \otimes f_y, P \rangle \equiv P_{x,y} = \langle m_x, n_y \rangle ~,
\]
with $\| m_x \|_2 = 1$ and $\| n_x \|_2 = 1$ for all $1 \leq x \leq |\mathcal{X}|$ and  $1 \leq y \leq |\mathcal{Y}|$, respectively.
Furthermore, the matrices $(m_x)$ and $(n_y)^T$ give a factorization of $\hat{P}$, i.e., we have $\hat{P} = (n_y)^T \cdot (m_x)$. Using the definition of the $\gamma_2$ tensor norm given by (\ref{eqnAltDefW2}) in Appendix \ref{sectionHilbertTensor1} (with $|\mathcal{A}| = |\mathcal{B}| = 1$) yields
\[
\gamma_2(P) \leq \| (n_y)^T \|_{2 \rightarrow \infty} \cdot \| (m_x) \|_{1 \rightarrow 2} ~. 
\]
By applying Lemma \ref{lemmaRepRSNormInf2} and using that $\| m_x \|_2 = 1$ and $\| n_x \|_2 = 1$, we get $\| (n_y)^T \|_{2 \rightarrow \infty} = \| (m_x) \|_{1 \rightarrow 2} = 1$ and hence, $\gamma_2(P) \leq 1$.

For the converse, assume that $\gamma_2(P) \leq 1$. Then, by the definition given in  (\ref{eqnAltDefW2}) we can conclude that there exist real vectors $\{ m_x \}$ and $\{ n_y \}$ such that $\| (n_y)^T \|_{2 \rightarrow \infty} \leq 1$ and $\| (m_x) \|_{1 \rightarrow 2} \leq 1$ with $\langle f_x \otimes f_y , P \rangle \equiv P_{x,y} = \langle m_x, n_y \rangle$. Then, the second part of Lemma \ref{lemmaRepRSNormInf2} implies that $\| m_x \|_2 \leq 1$ and $\| n_y \|_2 \leq 1$ for all $1 \leq x \leq |\mathcal{X}|$ and  $1 \leq y \leq |\mathcal{Y}|$. 
Applying the second part of Lemma \ref{theoremTsirelson} on the vectors $\{ m_x \}$ and $\{ n_y \}$ implies that $P$ is indeed a quantum correlation.
\end{proof}

\setcounter{proposition}{1}
We would also like to prove a similar result as given by Lemma \ref{lemmaQuantumCorrGamma} for quantum systems $P \in \ell_{\infty}^{|\mathcal{X} |}(\ell_{1}^{|\mathcal{A}|}) \otimes \ell_{\infty}^{|\mathcal{Y} |}(\ell_{1}^{|\mathcal{B}|})$. Unfortunately, no generalization of Tsirelson's theorem to many outputs is known to exist\footnote{although, one might consider the \emph{quantum rounding} method in \cite{uniqueGamesEasy} as some kind of approximated version of Tsirelson's theorem for unique games.}. In particular, the second part of Tsirelson's theorem is the problem, as the first part can be generalized as will be seen in the proof of Proposition 2. We therefore get a weaker result, namely
\begin{proposition}
\label{lemmaGammaQuantumState}
Let $P \in \ell_{\infty}^{|\mathcal{X} |}(\ell_{1}^{|\mathcal{A}|}) \otimes \ell_{\infty}^{|\mathcal{Y} |}(\ell_{1}^{|\mathcal{B}|})$ be a quantum system. Then $\gamma_2(P) = 1$.
\end{proposition}

\begin{proof}
As $P$ is a quantum system there exists a pure quantum state $| \Psi \rangle  \in \mathcal{H}_A \otimes \mathcal{H}_B$ and projective measurements $\{ M_x^a\}_{1 \leq a \leq |\mathcal{A}|}$ and $\{ N_y^b\}_{1 \leq b \leq |\mathcal{B}|}$ with $\sum_{a = 1}^{|\mathcal{A}|} M_x^a = id_{\mathcal{H}_A}$ and $\sum_{b = 1}^{|\mathcal{B}|} N_y^b = id_{\mathcal{H}_B}$, respectively, such that
\[
 \langle f_{x,a} \otimes f_{y,b}, P \rangle \equiv P_{x,y}^{a,b} = \langle \Psi | M_x^a \otimes N_y^b | \Psi \rangle = \langle \Psi | (M_x^a \otimes id_{\mathcal{H}_B}) \cdot (id_{\mathcal{H}_A} \otimes N_y^b) | \Psi \rangle ~,
\]
with $P = \sum_{x,y,a,b} P_{x,y}^{a,b} \cdot e_{x,a} \otimes e_{y,b}$. Let $\{ |i \rangle \}_{i}$ be an orthonormal basis of $\mathcal{H}_A \otimes \mathcal{H}_B$. We define the complex vectors $\tilde{m}_{x,a}$ and $\tilde{n}_{y,b}$ by 
\begin{eqnarray}
\tilde{m}_{x,a}^i &:=& \langle \Psi | (M_x^a \otimes id_{\mathcal{H}_B}) | i \rangle ~, \nonumber \\
\tilde{n}_{y,b}^i &:=& \langle i | (id_{\mathcal{H}_A} \otimes N_y^b) | \Psi \rangle ~, \nonumber
\end{eqnarray}
where $\tilde{m}_{x,a}^i$ indicates the $i$'th entry of the vector $\tilde{m}_{x,a}$. Hence,
\begin{equation}
\label{eqnSumRepVecPsi}
\sum_i \tilde{m}_{x,a}^i \cdot \tilde{n}_{y,b}^i = \langle \Psi | M_x^a \otimes N_y^b | \Psi \rangle ~,
\end{equation}
by using that $\sum_i |i \rangle \langle i | = id_{\mathcal{H}_A} \otimes id_{\mathcal{H}_B}$. Note that
\[
\| \tilde{m}_{x,a} \|_2^2 = \sum_i | \tilde{m}_{x,a}^i |^2 = \langle \Psi | (M_x^a \otimes id_{\mathcal{H}_B}) | \Psi \rangle ~,
\]
and therefore, that 
\begin{equation}
\label{eqnNormVec1AB}
\sum_{a = 1}^{|\mathcal{A}|} \| \tilde{m}_{x,a} \|_2^2 = \sum_{a = 1}^{|\mathcal{A}|} \langle \Psi | (M_x^a \otimes id_{\mathcal{H}_B}) | \Psi \rangle = \langle \Psi | \Psi \rangle = 1 ~,
\end{equation}
for all $x \in \{1,2,...,|\mathcal{X}| \}$. Similarly $\sum_{b = 1}^{|\mathcal{B}|} \| \tilde{n}_{y,b} \|_2^2 = 1$ for all $y \in \{1,2,...,|\mathcal{Y}| \}$.
Furthermore, the vectors $\{ \tilde{m}_{x,a} \}_a$ are mutually orthogonal for a given $x$, i.e., $\langle \tilde{m}_{x,a_1}, \tilde{m}_{x,a_2} \rangle = \delta_{a_1,a_2} \cdot \| \tilde{m}_{x,a_1} \|_2^2$, with $a_1, a_2 \in \{1,2,...,|\mathcal{A}| \}$, and for all $x \in \{1,2,...,|\mathcal{X}| \}$. This is the case since
\begin{eqnarray}
\label{eqnOrthogMxa1a2}
\langle \tilde{m}_{x,a_1}, \tilde{m}_{x,a_2} \rangle &=& \langle \Psi | (M_x^{a_1} \cdot M_x^{a_2} \otimes id_{\mathcal{H}_B}) | \Psi \rangle \nonumber \\
&=& \langle \Psi | (M_x^{a_1} \otimes id_{\mathcal{H}_B}) | \Psi \rangle \cdot \delta_{a_1,a_2} \nonumber \\
&=& \| \tilde{m}_{x,a_1} \|_2^2 \cdot \delta_{a_1,a_2} ~,
\end{eqnarray}
as $M_x^{a_1}$ and $M_x^{a_2}$ are projectors with the property that $M_x^{a_1} \cdot M_x^{a_2} = M_x^{a_1} \cdot \delta_{a_1,a_2}$. By an analogous argument one can show that also $\langle \tilde{n}_{y,b_1}, \tilde{n}_{y,b_2} \rangle = \delta_{b_1,b_2} \cdot \| \tilde{n}_{y,b_1} \|_2^2$, with $b_1, b_2 \in \{1,2,...,|\mathcal{B}| \}$, and for all $y \in \{1,2,...,|\mathcal{Y}| \}$.
Note that the complex vectors $\tilde{m}_{x,a}$ and $\tilde{n}_{y,b}$ can be replaced by real vectors $m_{x,a}$ and $n_{y,b}$ of twice the length while still fulfilling (\ref{eqnNormVec1AB}) and (\ref{eqnOrthogMxa1a2}).

In order to compute $\gamma_2$ we will use its version in  (\ref{eqnAltDefW2}). Let us define the matrices $R := (n_{y,b})^T$ and $S := (m_{x,a})$.
We then have that $\hat{P} = R \cdot S$. Let us now show that $\| R \|_{2 \rightarrow \infty(1)} = 1$ and  $\| S \|_{1(\infty) \rightarrow 2} = 1$ and therefore prove that $\gamma_2(P) \leq \| R \|_{2 \rightarrow \infty(1)} \cdot \| S \|_{1(\infty) \rightarrow 2} = 1$. That $\gamma_2(P) \geq 1$ follows form Lemma \ref{lemmaLowerAllNormMnorm}. By applying Lemma \ref{lemmaRepRSNormInf2} we get
\begin{eqnarray}
\| (m_{x,a}) \|_{1(\infty) \rightarrow 2}^2 &=& \max_{s, x} \left\| \sum_{a=1}^{|\mathcal{A}|} s(a) \cdot m_{x,a} \right\|_2^2 \nonumber \\
&=& \max_{s,x} \sum_{a} \langle m_{x,a}, m_{x,a} \rangle \nonumber \\
&=& \max_{s,x} \sum_{a} \| m_{x,a} \|_2^2 = 1 ~, \nonumber
\end{eqnarray}
where we used the orthogonality relations of  (\ref{eqnOrthogMxa1a2}) in the second line and  (\ref{eqnNormVec1AB}) in the third line. By the same argument one can show that $\| (n_{y,b})^T \|_{2 \rightarrow \infty(1)} = 1$ holds as well.
\end{proof}

\subsection{Dual of Hilbertian Tensor Norm and XOR Games}
\label{sectionDualHilXor}

\setcounter{proposition}{3}

\begin{proposition}
\label{lemmaOmegaGammaEquiv}
Let $G = (\pi,V)$ be an XOR game with $G \in \ell_{1}^{|\mathcal{X} |}(\ell_{\infty}^{|\mathcal{A}|}) \otimes \ell_{1}^{|\mathcal{Y} |}(\ell_{\infty}^{|\mathcal{B}|})$ and $|\mathcal{A}| = |\mathcal{B}| = 2$. Then
\[
\omega^*(G) = \gamma_2^*(G) ~.
\]
\end{proposition}

\begin{proof}
That $\omega^*(G) \leq \gamma_2^*(G)$ follows from Proposition \ref{lemmaUpperOmegaGamma}.
So let us show that $\omega^*(G) \geq \gamma_2^*(G)$.
Let $\mathcal{A} = \mathcal{B} = \{0,1\}$ and $P_{x,y}^{a,b} = \langle f_{x,a} \otimes f_{y,b}, P \rangle$ with $P \in \ell_{\infty}^{|\mathcal{X} |}(\ell_{1}^{|\mathcal{A}|}) \otimes \ell_{\infty}^{|\mathcal{Y} |}(\ell_{1}^{|\mathcal{B}|})$. Then
\begin{eqnarray}
\label{eqnGammaEqivUpperX}
\gamma_2^*(G) &=& \sup \{ | \langle G, P \rangle | ~:~ \gamma_2(P) \leq 1 \} \nonumber \\
&=& \sup \left\{  \sum_{x,y} \pi(x,y) \sum_{a,b} V(a,b,x,y) \cdot P_{x,y}^{a,b}  ~:~ \gamma_2(P) \leq 1 \right\} \nonumber \\
&=& \sup \left\{  \sum_{x,y} \pi(x,y) \left( V(a=b,x,y) \cdot ( P_{x,y}^{0,0} + P_{x,y}^{1,1}) \right. \right. \nonumber \\ 
&+& \left. \left. V(a \neq b, x,y) \cdot (P_{x,y}^{0,1} + P_{x,y}^{1,0}) \right)  ~:~ \gamma_2(P) \leq 1 \right\} ~, 
\end{eqnarray}
where we used the fact that $G$ is an XOR game, i.e., we have that $V(a=b,x,y) := V(0,0,x,y) = V(1,1,x,y) \in \{0,1\}$ and $V(a \neq b,x,y) := V(0,1,x,y) = V(1,0,x,y) \in \{0,1\}$. We do not have to take the absolute value as $\gamma_2(P) = \gamma_2((-1)\cdot P)$. Using the definition of $\gamma_2$ given in  (\ref{eqnAltDefW2}) and applying Lemma \ref{lemmaRepRSNormInf2}, from $\gamma_2(P) \leq 1$ we get the following constraints:
\begin{equation}
\label{eqnCondmn}
\| m_{x,0} \pm m_{x,1} \|_2 \leq 1 ~,~ \| n_{y,0} \pm n_{y,1} \|_2 \leq 1 ~,~ P_{x,y}^{a,b} = \langle m_{x,a}, n_{y,b} \rangle ~,
\end{equation}
for all $x \in \mathcal{X}$ and $y \in \mathcal{Y}$, respectively.
We cannot hope to simulate the ``distribution'' $P_{x,y}^{a,b}$ in  (\ref{eqnGammaEqivUpperX}) by some measurements on a quantum state, as these values can be negative and do not have to correspond to valid probabilities. But what we can do is to show that there exists a quantum state $| \Psi \rangle$ and observables $A_1,...,A_{|\mathcal{X}|}$ and $B_1,...,B_{|\mathcal{Y}|}$ with binary outcomes such that
\begin{equation}
\label{eqnConEquiv}
 \Pr[a = b | A_x,B_y, |\Psi\rangle ] \geq P_{x,y}^{0,0} + P_{x,y}^{1,1} ~,~ \Pr[a \neq b | A_x,B_y, |\Psi\rangle ] \geq P_{x,y}^{0,1} + P_{x,y}^{1,0} ~,
\end{equation}
where $a$ is the outcome of Alice's measurement $A_x$ and $b$ the outcome of Bob's measurement $B_y$. So, if we assume that  (\ref{eqnConEquiv}) holds, we get
\begin{eqnarray}
\gamma_2^*(G)&=& \sup \left\{  \sum_{x,y} \pi(x,y) ( V(a=b,x,y) \cdot ( P_{x,y}^{0,0} + P_{x,y}^{1,1}) \right. \nonumber \\ 
&+& \left. V(a \neq b, x,y) \cdot (P_{x,y}^{0,1} + P_{x,y}^{1,0}) )  ~:~ \gamma_2(P) \leq 1 \right\} \nonumber \\
&\leq&  \sum_{x,y} \pi(x,y) \left( V(a=b,x,y) \Pr[a = b | A_x,B_y, |\Psi\rangle ] \right. \nonumber \\
&+& \left. V(a \neq b,x,y) \Pr[a \neq b | A_x,B_y, |\Psi\rangle ] \right)  \nonumber \\
&\leq& \omega^*(G) ~. \nonumber
\end{eqnarray}
It remains to be shown that  (\ref{eqnConEquiv}) can be achieved. First, note that  (\ref{eqnConEquiv}) can be rewritten as
\begin{eqnarray}
\label{eqnRepDifPer}
\Pr[a = b | A_x,B_y, |\Psi\rangle ] &\geq& \langle m_{x,0}, n_{y,0} \rangle + \langle m_{x,1}, n_{y,1} \rangle ~,  \\
\Pr[a \neq b | A_x,B_y, |\Psi\rangle ] &\geq& \langle m_{x,0}, n_{y,1} \rangle + \langle m_{x,1}, n_{y,0} \rangle ~,
\end{eqnarray}
by using  (\ref{eqnCondmn}).
And second, we set $m_x := m_{x,0} - m_{x,1}$ and $n_y := n_{y,0} - n_{y,1}$, apply the second part of Lemma \ref{theoremTsirelson} (which we are allowed to use because of the constraints given in  (\ref{eqnCondmn})) and get observables $A_1,...,A_{|\mathcal{X}|}$ and $B_1,...,B_{|\mathcal{Y}|}$ with eigenvalue $\pm 1$ and a quantum state $| \Psi \rangle$ such that 
\[
\langle m_x, n_y \rangle = \langle \Psi | A_x \otimes B_y | \Psi \rangle ~.
\]
As $\langle \Psi | A_x \otimes B_y | \Psi \rangle$ is the expectation value when measuring the observables $A_x$ and $B_y$ with eigenvalues $\pm 1$, we have that
\begin{eqnarray}
\label{eanREpPRF}
\Pr[a = b | A_x,B_y, |\Psi\rangle ] &=& \frac{1 + \langle \Psi | A_x \otimes B_y | \Psi \rangle}{2} = \frac{1 + \langle m_x, n_y \rangle}{2}  ~, \\
\Pr[a \neq b | A_x,B_y, |\Psi\rangle ] &=& \frac{1 - \langle \Psi | A_x \otimes B_y | \Psi \rangle}{2} = \frac{1 - \langle m_x, n_y \rangle}{2} ~.
\end{eqnarray}
By straightforward calculations,  (\ref{eanREpPRF}) implies  (\ref{eqnRepDifPer}), where the conditions $\| m_{x,0} + m_{x,1} \|_2 \leq 1$ and $\| n_{y,0} + n_{y,1} \|_2 \leq 1$ of  (\ref{eqnCondmn}) are used. And similarly for $\Pr[a \neq b | A_x,B_y, |\Psi\rangle ]$.
\end{proof}

\appendix

\renewcommand{\theequation}{A-\arabic{equation}}
\setcounter{equation}{0}

\section*{APPENDIX}

\section{Basic Properties of Banach Spaces and Tensor Products}
\label{sectionNotation}

We call $\| \cdot \|_X : V \rightarrow \mathbb{R}^+_0$ a \emph{norm} over the vector space $V$ if it fulfils the following three conditions:
\begin{enumerate}
 \item $\| v \|_X = 0$ if and only if $v = 0$.
 \item $\| c \cdot v \|_X = |c| \cdot \| v \|_X$, for all $c \in \mathbb{R}$ and $v \in V$.
 \item $\| v + w \|_X \leq \| v \|_X + \| w \|_X$, for all $v,w \in V$.
\end{enumerate}
Given a vector space $V$ and a norm $\| \cdot \|_X$ on it, the tuple $X = (V, \| \cdot \|_X)$ is called a \emph{normed space}. 
A normed space $X = (V, \| \cdot \|_X)$ with $V$ finite dimensional is also a \emph{Banach space}. From now on let $V$ always be finite dimensional and therefore we consider only finite dimensional Banach spaces.

\begin{lemma}
\label{lemmaTakingDuals}
Let $X := (\mathbb{R}^n, \| \cdot \|_X)$ and $Y := (\mathbb{R}^n, \| \cdot \|_Y)$ be a Banach spaces. Then, for any positive constant $c \in \mathbb{R}$,
\[
\| P \|_X  \leq c \cdot \| P \|_Y  ~,\forall P \in \mathbb{R}^n ~ \Leftrightarrow \| G \|_{Y^*} \leq c \cdot \| G \|_{X^*} ~, \forall G \in \mathbb{R}^n ~.  
\]
\end{lemma}

\begin{proof}
By the definition of the dual norm given in (\ref{eqnDefDualNorm}) we obtain
\begin{eqnarray}
\| G \|_{Y^*} &=& \sup \{ | \langle G, P \rangle | ~:~ \| P \|_Y \leq 1 \} \nonumber \\
&=& c \cdot \sup \{ | \langle G, P \rangle | ~:~ c \cdot \| P \|_Y \leq 1 \} \nonumber \\
&\leq& c \cdot \sup \{ | \langle G, P \rangle | ~:~ \| P \|_X \leq 1 \} \nonumber \\
&=& c \cdot \| G \|_{X^*} ~, \nonumber
\end{eqnarray}
and similarly for the other direction.
\end{proof}

An example of a Banach space is $\ell_2^n := (\mathbb{R}^n, \| \cdot \|_2)$ with the norm defined as
\begin{equation}
\label{eqnDef2norm}
\| v \|_2 := \left( \sum_{i=1}^n | \langle f_i, v \rangle |^2 \right)^{1/2} ~.
\end{equation}
Note that we write $\| \cdot \|_2$ instead of $\| \cdot \|_{\ell_2^n}$ in order to simplify the notation. The special properties the space $\ell_2^n$ has make it to a \emph{Hilbert space}. A Hilbert space $\mathcal{H}$ is a Banach space where the norm fulfils the parallelogram identity, i.e.,
\[
\| v + w \|_{\mathcal{H}}^2 + \| v - w \|_{\mathcal{H}}^2 = 2 \| v \|_{\mathcal{H}}^2 + 2 \| w \|_{\mathcal{H}}^2 ~,
\]
for all $v,w \in \mathcal{H}$. Furthermore, in that case the norm $\| \cdot \|_{\mathcal{H}}$ uniquely induces an \emph{inner product} on the Hilbert space. The Banach space $\ell_2^n$ has the nice property that it is \emph{self dual}, i.e., $\ell_2^n \cong (\ell_2^n)^*$ which means that 
\[
\| v \|_{2^*} = \sup_{w \in \mathbb{R}^n} \{ | \langle v, w \rangle | ~:~ \| w \|_2 \leq 1 \} = \| v \|_2  ~,
\]
for all $v \in (\mathbb{R}^n)^* \cong \mathbb{R}^n$. Note that in particular $\| v \|_2^2 = | \langle v, v \rangle |$.

\begin{lemma}
\label{lemmaDualityA}
Let $X := (\mathbb{R}^n, \| \cdot \|_X)$, with $1 \leq n < \infty$, be a Banach space with $\| \cdot \|_X$ an arbitrary norm over $\mathbb{R}^n$ and $A: \ell_2^m \rightarrow X$, for $1 \leq m \leq \infty$, a linear operator. Then
\[
\| A \|_{2 \rightarrow X} = \| A^T \|_{X^* \rightarrow 2} ~.
\]
 
\end{lemma}

\begin{proof}
Representing $A$ as a row matrix, with rows $a_i \in \ell_2^m$, for $1 \leq i \leq n$, yields
\begin{eqnarray}
\| A \|_{2 \rightarrow X} &=& \sup_{ \| \lambda \|_2 \leq 1 } \| A(\lambda) \|_X \nonumber \\
&=& \sup_{ \| \lambda \|_2 \leq 1 } \left\| \left( \begin{array}{c} \langle a_1, \lambda \rangle \\ \vdots \\ \langle a_n, \lambda \rangle \end{array} \right) \right\|_X \nonumber \\
&=& \sup_{ \| \lambda \|_2 \leq 1 } \sup_{ \| \mu \|_{X^*} \leq 1} \left| \sum_{i=1}^n \mu_i \cdot \langle a_i, \lambda \rangle \right| ~.
\end{eqnarray}
On the other hand, by using that the $2$-norm is self dual, we obtain 
\begin{eqnarray}
\| A^T \|_{X^* \rightarrow 2} &=& \sup_{ \| \mu \|_{X^*} \leq 1} \left\| \sum_{i=1}^n \mu_i \cdot a_i \right\|_2 \nonumber \\
&=& \sup_{ \| \mu \|_{X^*} \leq 1} \sup_{ \| \lambda \|_2 \leq 1 } \left| \sum_{i=1}^n \mu_i \cdot \langle a_i, \lambda \rangle \right| ~,
\end{eqnarray}
and therefore $\| A \|_{2 \rightarrow X} = \| A^T \|_{X^* \rightarrow 2}$.
\end{proof}

The algebraic tensor product of two Banach spaces $X$ and $Y$ is denoted by $X \otimes Y$. Note that this is not yet a Banach space as we have not yet defined a norm on the tensor space $X \otimes Y$. An element $v \in X \otimes Y$ can always be written as
\[
v = \sum_{i = 1}^{m} v_A^i \otimes v_B^i ~,
\]
with $v_A^i \in X$ and $v_B^i \in Y$, respectively. It is important to note that this decomposition is not unique as, typically, there are infinitely many such representations. The tensor product has the following properties:
\begin{enumerate}
 \item $(v_1 + v_2) \otimes w = v_1 \otimes w + v_2 \otimes w$,
 \item $v \otimes (w_1 + w_2) = v \otimes w_1 + v \otimes w_2$,
 \item $c \cdot (v \otimes w) = (c \cdot v) \otimes w = v \otimes (c \cdot w)$,
 \item $0 \otimes w = v \otimes 0 = 0$,
\end{enumerate}
for all $v,v_1,v_2 \in X$, $w,w_1,w_2 \in Y$, and $c \in \mathbb{R}$. Furthermore, let $X$ and $Y$ be Banach spaces and $X^*$ and $Y^*$ their corresponding dual Banach spaces. Then, it holds that
\[
\langle w_A \otimes w_B, v_A \otimes v_B \rangle = \langle w_A, v_A \rangle \cdot \langle w_B, v_B \rangle ~, 
\]
with $w_A \in X^*$,$w_B \in Y^*$,$v_A \in X$ and $v_B \in Y$ and $w_A \otimes w_B \in X^* \otimes Y^*$ and $v_A \otimes v_B \in X \otimes Y$. In particular, given $v = \sum_i v_A^i \otimes v_B^i \in X \otimes Y$ and $w = \sum_j w_A^j \otimes w_B^j \in X^* \otimes Y^*$, we get
\[
\langle w, v \rangle = \sum_{i,j} \langle w_A^j \otimes w_B^j, v_A^i \otimes v_B^i \rangle = \sum_{i,j}  \langle w_A^j, v_A^i \rangle \cdot \langle w_B^j, v_B^i \rangle ~. 
\]
Furthermore, if the vector spaces are isomorphic to $\mathbb{R}^n$ and $\mathbb{R}^m$ for some $1 \leq n,m < \infty$, and $v,w \in \mathbb{R}^n \otimes \mathbb{R}^m$, then
\[
\langle w, v \rangle = \sum_{i=1}^n \sum_{j=1}^m \langle w, e_i \otimes e_j \rangle \cdot  \langle f_i \otimes f_j, v \rangle ~.
\]
We will also use the notation $e_{i,j} := e_i \otimes e_j$ and $f_{i,j} := f_i \otimes f_j$. 

The $2$-norm behaves 'nicely' on product tensors, i.e., it is easy to see that 
\begin{equation}
\label{eqnNiceTwonorm}
\| v \otimes w \|_2 = \| v \|_2 \cdot \| w \|_2 ~,
\end{equation}
for all $v \in \ell_2^n$ and $w \in \ell_2^m$.

\renewcommand{\theequation}{B-\arabic{equation}}
\setcounter{equation}{0}  

\section{Basic Properties of Tensor Norms}
\label{sectionIntroTensorNorm}

We will introduce the notion of tensor norms in this section. Let $X$ and $Y$ be arbitrary Banach spaces and $X \otimes Y$ the algebraic tensor product of these two spaces. We will call $X$ and $Y$ \emph{local} spaces. A \emph{tensor norm} is a norm on $X \otimes Y$, which is based on the local Banach spaces $X$ and $Y$ with some additional special properties. See also \cite{tensorNormsOperatorIdeals, introTensor} which give a good introduction to the subject of tensor norms. Recall that by $\alpha^*$ we denote the dual tensor norm of $\alpha$ in the sense of (\ref{eqnDefDualNorm}). The definition of tensor norms reads then as follows \cite{crossSpaces}:
\begin{definition}[Tensor Norm]
Let $X$ and $Y$ be finite dimensional Banach spaces. A norm $\alpha$ on $X \otimes Y$ is called a tensor norm if the following three conditions are satisfied:
\begin{enumerate}
 \item $\alpha(P_A \otimes P_B) = \| P_A \|_X \cdot \| P_B \|_Y$ for every $P_A \in X$ and $P_B \in Y$.
 \item $\alpha^*(G_A \otimes G_B) = \| G_A \|_{X^*} \cdot \| G_B \|_{Y^*}$ for every $G_A \in X^*$ and $G_B \in Y^*$.
 \item $\| T_A \otimes T_B \|_{X \otimes_{\alpha} Y \rightarrow X \otimes_{\alpha} Y} = \| T_A \|_{X \rightarrow X} \cdot \| T_B \|_{Y \rightarrow Y}$ for all linear maps $T_A: X \rightarrow X$ and $T_B: Y \rightarrow Y$.
\end{enumerate}
\end{definition}
If a norm on $X \otimes Y$ fulfils the first two conditions it is called a \emph{reasonable cross norm} \cite{introTensor}.
Note that we need only the first two properties of tensor norms in this paper and that in \cite{crossSpaces, introTensor} $\alpha$ with these three properties is called a \emph{uniform cross norm} whereas in \cite{tensorNormsOperatorIdeals} it is called a \emph{tensor norm}.

\subsection{Four Different Tensor Norms}
\label{sectionFourTensorNorms}

In the following we will define four different tensor norms (actually just two, but taking the duals gives us four). We will only write down the definitions for the case where the local Banach spaces $X$ and $Y$ are $\ell_{\infty}^{|\mathcal{X}|}(\ell_{1}^{|\mathcal{A}|})$ (or $\ell_{1}^{|\mathcal{X}|}(\ell_{\infty}^{|\mathcal{A}|})$) and $\ell_{\infty}^{|\mathcal{Y}|}(\ell_{1}^{|\mathcal{B}|})$ (or $\ell_{1}^{|\mathcal{Y}|}(\ell_{\infty}^{|\mathcal{B}|})$), respectively. Note, that because $|\mathcal{X}|$, $|\mathcal{Y}|$, $|\mathcal{A}|$ and $|\mathcal{B}|$ are finite, the resulting Banach spaces are finite dimensional as well.

\subsubsection{Projective and Injective Tensor Norm}
\label{sectionProjInjNorm}

The first two tensor norms, called the projective and injective tensor norm, are the ``extremal`` ones., i.e., all tensor norms are larger than the injective and smaller than the projective tensor norm.
The projective tensor norm of $P \in \ell_{\infty}^{|\mathcal{X}|}(\ell_{1}^{|\mathcal{A}|}) \otimes \ell_{\infty}^{|\mathcal{Y}|}(\ell_{1}^{|\mathcal{B}|})$ is defined by
\[
\pi(P) := \inf \left\{ \sum_{i=1}^n \| P_A^i \|_{\infty(1)} \cdot \| P_B^i \|_{\infty(1)} ~:~ P = \sum_{i=1}^n P_A^i \otimes P_B^i  \right\} ~,
\]
where the infimum is over all decompositions (or representations) of $P$.
The injective tensor norm of $G \in \ell_{1}^{|\mathcal{X}|}(\ell_{\infty}^{|\mathcal{A}|}) \otimes \ell_{1}^{|\mathcal{Y}|}(\ell_{\infty}^{|\mathcal{B}|})$ is defined by
\[
\varepsilon(G) := \sup \left\{ | \langle G , P_A \otimes P_B \rangle | ~:~ \| P_A \|_{\infty(1)} \leq 1, \| P_B \|_{\infty(1)} \leq 1 \right\} ~,
\]
where the supremum is over $P_A \in \ell_{\infty}^{|\mathcal{X}|}(\ell_{1}^{|\mathcal{A}|})$ and $P_B \in \ell_{\infty}^{|\mathcal{Y}|}(\ell_{1}^{|\mathcal{B}|})$.

One can show \cite{introTensor} that these two norms are the dual of each other, i.e., 
\begin{eqnarray}
\pi(P) &=& \sup \{ | \langle G, P \rangle | ~:~ \varepsilon(G) \leq 1 \} ~, \nonumber \\
\varepsilon(G) &=& \sup \{ | \langle G, P \rangle | ~:~ \pi(P) \leq 1 \} ~, \nonumber
\end{eqnarray}
for $P \in \ell_{\infty}^{|\mathcal{X}|}(\ell_{1}^{|\mathcal{A}|}) \otimes \ell_{\infty}^{|\mathcal{Y}|}(\ell_{1}^{|\mathcal{B}|})$ and $G \in \ell_{1}^{|\mathcal{X}|}(\ell_{\infty}^{|\mathcal{A}|}) \otimes \ell_{1}^{|\mathcal{Y}|}(\ell_{\infty}^{|\mathcal{B}|})$.

As already state above, these two tensor norms are extremal. Formally, we have \cite{introTensor}:
\begin{lemma}
\label{lemmaCrossnorm}
Let $X$ and $Y$ be Banach spaces. Every tensor norm $\alpha$ on $X \otimes Y$ satisfies
\[
\varepsilon(P) \leq \alpha(P) \leq \pi(P)
\]
for every $P \in X \otimes Y$, where $\varepsilon$ is the injective tensor norm and $\pi$ is the projective tensor norm, both defined over $X \otimes Y$.
\end{lemma}

The next lemma states that if $P \in \ell_{\infty}^{|\mathcal{X}|}(\ell_{1}^{|\mathcal{A}|}) \otimes \ell_{\infty}^{|\mathcal{Y}|}(\ell_{1}^{|\mathcal{B}|})$ corresponds to an (almost) valid conditional probability distribution (we allow negative entries), then all tensor norms will assign to $P$ a value which is at least one.
\begin{lemma}
\label{lemmaLowerAllNormMnorm}
Let $P \in \ell_{\infty}^{|\mathcal{X}|}(\ell_{1}^{|\mathcal{A}|}) \otimes \ell_{\infty}^{|\mathcal{Y}|}(\ell_{1}^{|\mathcal{B}|})$ with $\sum_{a,b} \langle f_{x,a} \otimes f_{y,b}, P \rangle = 1$ for all $1 \leq x \leq | \mathcal{X} |$ and $1 \leq y \leq | \mathcal{Y} |$. Then
\[
\alpha(P) \geq 1 ~,
\]
for all tensor norms $\alpha$ over $\ell_{\infty}^{|\mathcal{X}|}(\ell_{1}^{|\mathcal{A}|}) \otimes \ell_{\infty}^{|\mathcal{Y}|}(\ell_{1}^{|\mathcal{B}|})$.
\end{lemma}

\begin{proof}
By the definition of the injective tensor norm we have
\begin{eqnarray}
\varepsilon(P) &=& \sup \left\{ \left| \langle G_A \otimes G_B, P \rangle \right| ~:~ \| G_A \|_{1(\infty)} \leq 1, \| G_B \|_{1(\infty)} \leq 1 \right\} \nonumber \\
&\geq& | \langle \mathbb{I}_A \otimes \mathbb{I}_B, P \rangle | ~, \nonumber
\end{eqnarray}
where $\mathbb{I}_A$ is the all-$1$ vector multiplied by $1 / |\mathcal{X}|$ and $\mathbb{I}_B$ is the all-$1$ vector multiplied by $1 / |\mathcal{Y}|$, where $\| \mathbb{I}_A \|_{1(\infty)} = 1$ and $\| \mathbb{I}_B \|_{1(\infty)} = 1$, respectively.
Taking the tensor product of $\mathbb{I}_A$ and $\mathbb{I}_B$ yields the all-$1$ vector multiplied by $1 / (|\mathcal{X}| |\mathcal{Y}|)$. Then, by using $\sum_{a,b} \langle f_{x,a} \otimes f_{y,b}, P \rangle = 1$, we obtain
\[
| \langle \mathbb{I}_A \otimes \mathbb{I}_B, P \rangle | = \frac{1}{|\mathcal{X}| |\mathcal{Y}|} \cdot |\mathcal{X}| |\mathcal{Y}| = 1 ~.
\]
Applying Lemma \ref{lemmaCrossnorm} results in
\[
\alpha(P) \geq \varepsilon(P) \geq 1 ~,
\]
for all tensor norms $\alpha$.
\end{proof}

\subsubsection{Hilbertian Tensor Norm and its Dual}
\label{sectionHilbertTensor1}

In this section we introduce the Hilbertian tensor norm, denoted by $\gamma_2$. One possible way to define it is \cite{tensorNormsOperatorIdeals}:
\begin{equation}
\label{defGamma2First}
 \gamma_2(P) := \inf \left\{ w_2(P_A^i ; \ell_{\infty}^{|\mathcal{X}|}(\ell_{1}^{|\mathcal{A}|})) \cdot w_2(P_B^i ; \ell_{\infty}^{|\mathcal{Y}|}(\ell_{1}^{|\mathcal{B}|})) ~:~  P = \sum_{i=1}^n P_A^i \otimes P_B^i \right\} ~,
\end{equation}
where the infimum is over all decomposition $P = \sum_i P_A^i \otimes P_B^i \in \ell_{\infty}^{|\mathcal{X}|}(\ell_{1}^{|\mathcal{A}|}) \otimes \ell_{\infty}^{|\mathcal{Y}|}(\ell_{1}^{|\mathcal{B}|})$ and
\[
w_2(P_A^i ; X) :=  \sup_{ \| G_A \|_{X^*} \leq 1} \left( \sum_{i=1}^n | \langle G_A, P_A^i \rangle |^{2} \right)^{1/2} ~.
\]
On the other hand, if $\gamma_2$ is defined over $\ell_{1}^{|\mathcal{X}|}(\ell_{\infty}^{|\mathcal{A}|}) \otimes \ell_{1}^{|\mathcal{Y}|}(\ell_{\infty}^{|\mathcal{B}|})$  we set:
\begin{equation}
\label{defGamma2First1}
 \gamma_2(G) := \inf \left\{ w_2(G_A^i ; \ell_{1}^{|\mathcal{X}|}(\ell_{\infty}^{|\mathcal{A}|})) \cdot w_2(G_B^i ; \ell_{1}^{|\mathcal{Y}|}(\ell_{\infty}^{|\mathcal{B}|})) ~:~  G = \sum_{i=1}^n G_A^i \otimes G_B^i \right\} ~,
\end{equation}
where the infimum is over all decomposition $G = \sum_i G_A^i \otimes G_B^i \in \ell_{1}^{|\mathcal{X}|}(\ell_{\infty}^{|\mathcal{A}|}) \otimes \ell_{1}^{|\mathcal{Y}|}(\ell_{\infty}^{|\mathcal{B}|})$.

The dual of $\gamma_2$ can be represented by \cite{tensorNormsOperatorIdeals}:
\begin{equation}
\label{eqnDefGammaDual}
 \gamma_2^{*}(G) := \inf \left\{  \| (\mu_{ij}) \|_{2 \rightarrow 2} \cdot \ell_2(G_A^i ; \ell_{1}^{|\mathcal{X}|}(\ell_{\infty}^{|\mathcal{A}|}) ) \cdot \ell_2(G_B^j ; \ell_{1}^{|\mathcal{Y}|}(\ell_{\infty}^{|\mathcal{B}|}) ) \right\}  ~,
\end{equation}
where the infimum is over all decompositions $G = \sum_{i,j}^n \mu_{ij} \cdot G_A^i \otimes G_B^j \in \ell_{1}^{|\mathcal{X}|}(\ell_{\infty}^{|\mathcal{A}|}) \otimes \ell_{1}^{|\mathcal{Y}|}(\ell_{\infty}^{|\mathcal{B}|})$, $(\mu_{ij})$ is a real $n \times n$-matrix, and 
\[
\ell_2(G_A^i ; X) :=  \left( \sum_{i=1}^n \| G_A^i \|^2_{X} \right)^{1/2} ~.
\]
And similarly for $\gamma_2^*$ over the tensor space $\ell_{\infty}^{|\mathcal{X}|}(\ell_{1}^{|\mathcal{A}|}) \otimes \ell_{\infty}^{|\mathcal{Y}|}(\ell_{1}^{|\mathcal{B}|})$ we set
\begin{equation}
\label{eqnDefGammaDual1}
 \gamma_2^{*}(P) := \inf \left\{ \| (\mu_{ij}) \|_{2 \rightarrow 2} \cdot \ell_2(P_A^i ; \ell_{\infty}^{|\mathcal{X}|}(\ell_{1}^{|\mathcal{A}|}) ) \cdot \ell_2(P_B^j ; \ell_{\infty}^{|\mathcal{Y}|}(\ell_{1}^{|\mathcal{B}|}) ) \right\}  ~,
\end{equation}
where the infimum is over all decompositions $P = \sum_{i,j}^n \mu_{ij} \cdot P_A^i \otimes P_B^j \in \ell_{\infty}^{|\mathcal{X}|}(\ell_{1}^{|\mathcal{A}|}) \otimes \ell_{\infty}^{|\mathcal{Y}|}(\ell_{1}^{|\mathcal{B}|})$.

There is a useful alternative representation of the Hilbertian tensor norm $\gamma_2$ whereof it actually got its name from.
A tensor $P \in \ell_{\infty}^{|\mathcal{X}|}(\ell_{1}^{|\mathcal{A}|}) \otimes \ell_{\infty}^{|\mathcal{Y} |}(\ell_{1}^{|\mathcal{B}|})$ can be interpreted as a linear operator $\hat{P}: (\ell_{\infty}^{|\mathcal{X} |}(\ell_{1}^{|\mathcal{A}|}))^* \rightarrow \ell_{\infty}^{|\mathcal{Y} |}(\ell_{1}^{|\mathcal{B}|})$ by the following identification:
\begin{equation}
\label{eqnDefRhoOp}
\hat{P}(G) := \sum_{i} \langle G, P_A^i \rangle \cdot P_B^i ~, 
\end{equation}
with $P = \sum_{i} P_A^i \otimes P_B^i$ and $G \in (\ell_{\infty}^{|\mathcal{X}|}(\ell_{1}^{|\mathcal{A}|}))^* \cong \ell_{1}^{|\mathcal{X} |}(\ell_{\infty}^{|\mathcal{A}|})$. Note that $\hat{P}(G)$ does not depend on the actual decomposition of $P$. We are now ready to state the alternative representation of the $\gamma_2$ norm \cite{introTensor}, namely:
\begin{equation}
\label{eqnAltDefW2}
\gamma_2(P) = \inf_{ \hat{P} = R \cdot S} \| R \|_{2 \rightarrow \infty(1)} \cdot \| S \|_{1(\infty) \rightarrow 2} ~,
\end{equation}
where the infimum is over all decomposition of $\hat{P}$ into linear operators $R: \ell_2 \rightarrow \ell_{\infty}^{|\mathcal{Y} |}(\ell_{1}^{|\mathcal{B}|})$ and $S : \ell_{1}^{|\mathcal{X} |}(\ell_{\infty}^{|\mathcal{A}|}) \rightarrow \ell_2$. In other words, $\hat{P}$ is factored through the Hilbert space $\ell_2$. Note that this Hilbert space can be of any dimension, even infinite dimensional. By setting $|\mathcal{A}| = |\mathcal{B}| = 1$ we recover the norms used in \cite{complexMeasure, lowerBoundsComm, directProductTheorems}. See Appendix \ref{sectionEquivGamma} for a proof of this equivalence.

We can think of $R$ and $S$ being matrices of dimension $|\mathcal{Y} | |\mathcal{B}| \times n$ and $n \times |\mathcal{X}| |\mathcal{A}|$, respectively, with $1 \leq n \leq \infty$, such that their matrix product yields $\hat{P}$. Representing $R$ as a row matrix $R := (n_{y,b})^T$ and $S$ as a column matrix $S := (m_{x,a})$ (see also Section \ref{sectionNotationProofs} about the notation) yields as entries of $\hat{P} = R \cdot S$ the values $\langle f_{x,a} \otimes f_{y,b}, P \rangle = \langle n_{y,b}, m_{x,a} \rangle = \langle m_{x,a}, n_{y,b} \rangle$. 
An immediate corollary is
\begin{corollary}
\label{corollaryEqivGammaVec}
Let $P \in \ell_{\infty}^{|\mathcal{X}|}(\ell_{1}^{|\mathcal{A}|}) \otimes \ell_{\infty}^{|\mathcal{Y} |}(\ell_{1}^{|\mathcal{B}|})$. Then $\gamma_2(P) \leq 1$ if and only if there exist vectors $m_{x,a}, n_{y,b} \in \ell_2^n$, with $1 \leq n \leq \infty$, such that $\langle f_{x,a} \otimes f_{y,b}, P \rangle = \langle m_{x,a}, n_{y,b} \rangle$, and $\| (m_{x,a}) \|_{1(\infty) \rightarrow 2} \leq 1$ and $\| (n_{y,b})^T \|_{2 \rightarrow \infty(1)} \leq 1$.
\end{corollary}

\renewcommand{\theequation}{C-\arabic{equation}}
\setcounter{equation}{0}  

\section{Introduction to Bell Inequalities}
\label{sectionIntroBell}

A \emph{Bell inequality} $G :  \ell_{\infty}^{|\mathcal{X} |}(\ell_{1}^{|\mathcal{A}|}) \otimes \ell_{\infty}^{|\mathcal{Y} |}(\ell_{1}^{|\mathcal{B}|}) \rightarrow \mathbb{R}$ can be interpreted as a linear functional from the space of conditional probability distributions to the real numbers. Let us denote by $B_C(G)$ the maximal value that can be achieved by applying the Bell inequality $G$ on a \emph{classical} conditional probability distribution, i.e., 
\[
B_C(G) := \sup_P \{ | \langle G, P \rangle | ~:~ P \textit{~is~classical} \} ~,
\]
where by ``$P$ \emph{is classical}'' we mean that $P$ can be written as
\[
\langle f_{x,a} \otimes f_{y,b} , P \rangle = \int \rho(\lambda) \cdot P_{A | X \Lambda}(a , x, \lambda) \cdot P_{B | Y \Lambda}(b , y, \lambda) d\lambda
\]
with $P_{A | X \Lambda}(a , x, \lambda) \geq 0$, $P_{B | Y \Lambda}(b , y, \lambda) \geq 0$, $\sum_a P_{A | X \Lambda}(a , x, \lambda) = 1$, $\sum_b P_{B | Y \Lambda}(b , y, \lambda) = 1$ and $\int \rho(\lambda) d\lambda = 1$,
i.e., the distribution $P$ can be explained by a local hidden variable model, where the local hidden variable $\lambda$ is selected with probability $\rho(\lambda)$. Hence, by $\langle f_{x,a} \otimes f_{y,b} , P \rangle$ we refer to the probability that the outputs are $a$ and $b$, given the inputs $x$ and $y$. We say that a Bell inequality $G \in \ell_{1}^{|\mathcal{X} |}(\ell_{\infty}^{|\mathcal{A}|}) \otimes \ell_{1}^{|\mathcal{Y} |}(\ell_{\infty}^{|\mathcal{B}|})$ is violated by the conditional probability distribution $P$ if 
\[
\arrowvert \langle G, P \rangle | \equiv \left| \sum_{x,y,a,b} \langle G, e_{x,a} \otimes e_{y,b} \rangle \cdot \langle f_{x,a} \otimes f_{y,b} , P \rangle \right| > B_C(G) ~.
\]

The most prominent example of a Bell inequality is the so-called CHSH Bell inequality \cite{CHSH}. Let $\mathcal{A} = \mathcal{B} = \mathcal{X} = \mathcal{Y} = \{0 , 1\}$, i.e., there are only two inputs and two outputs on each side, respectively. The CHSH inequality is usually stated in the form of expectation values, but in order to fit into our presentation, we will state its equivalent ``probability representation'':
\begin{equation}
\label{eqnDefCHSH}
\langle G_{CHSH}, e_{x,a} \otimes e_{y,b} \rangle := \left\{ \begin{array}{l} +1 ~,~ \textit{if}~ a \oplus b = x \wedge y \\ -1 ~,~ \textit{otherwise} \end{array} \right. ~.
\end{equation}
It is not hard to show that $B_C(G_{CHSH}) = 2$, where this value can be achieved for $P$ which always ``outputs'' the values $a = 0$ and $b = 0$, independently of the inputs $x$ and $y$.

In Section \ref{sectionInjecProjecNorm}, we have shown that, for a game $G \in \ell_{1}^{|\mathcal{X} |}(\ell_{\infty}^{|\mathcal{A}|}) \otimes \ell_{1}^{|\mathcal{Y} |}(\ell_{\infty}^{|\mathcal{B}|})$, the injective tensor norm and classical value of the game are equal, i.e., that $\varepsilon(G) = \omega(G)$ (see Proposition \ref{lemmaEqiuvClassVAr}). For Bell inequalities $G$ there is no equality relation any more. It only holds that $B_C(G) \leq \varepsilon(G)$ for all Bell inequalities $G$.
The reason for losing the equality stems from the fact that, in contrast to two-prover games, a Bell inequality can have \emph{negative} entries. Furthermore, the fact that $B_C(G)$ is not equal to $\varepsilon(G)$ for Bell inequalities $G$ is the reason for our proof of Theorem \ref{theoremBellupperNew} not going through for Bell inequalities.

\renewcommand{\theequation}{D-\arabic{equation}}
\setcounter{equation}{0}  

\section{Equivalence of $\gamma_2$ Definitions}
\label{sectionEquivGamma}

We will show the following equality:
\begin{equation}
\label{eqnEqivNormGamma2}
\inf_{ \hat{P} = R \cdot S} \| R \|_{2 \rightarrow Y} \cdot \| S \|_{X^* \rightarrow 2} = \inf  w_2(P_A^i ; X ) \cdot w_2(P_B^i ; Y ) ~,
\end{equation}
with $P = \sum_{i=1}^{n} P_A^i \otimes P_B^i \in X \otimes Y$ for $X$ and $Y$ arbitrary finite dimensional Banach spaces, which implies the equivalence of (\ref{defGamma2First}) and (\ref{eqnAltDefW2}) in Appendix \ref{sectionHilbertTensor1}.

Let us first show that the right hand side of  (\ref{eqnEqivNormGamma2}) is larger or equal to the left hand side.
First, let $P = \sum_{i=1}^{n} P_A^i \otimes P_B^i \in X \otimes Y$ be the optimal decomposition on the right hand side of  (\ref{eqnEqivNormGamma2}). Then, we define $R : \ell_2^n \rightarrow Y$ and $S: X^* \rightarrow \ell_2^n$ as follows:
\begin{eqnarray}
R(\lambda) &:=& \sum_{i=1}^{n} \langle \lambda, e_i \rangle \cdot P_{B}^i ~, \nonumber \\
S(G_A) &:=& \sum_{i=1}^n e_i \cdot \langle G_A, P_A^i \rangle ~. \nonumber
\end{eqnarray}
The operator $\hat{P} : X^{*} \rightarrow Y$ corresponding to $P$ can be represented as 
\begin{equation}
\label{eqnRepP}
\hat{P}(G_A) = \sum_{i=1}^n \langle G_A, P_A^i \rangle \cdot P_B^i ~. 
\end{equation}
That $\hat{P} = R \cdot S$ indeed holds follows then by
\begin{equation}
\label{eqnEqivRSP}
 (R \cdot S)(G_A) = R(S(G_A)) = R\left( \sum_i e_i \cdot \langle G_A, P_A^i \rangle \right) = \sum_i \langle G_A, P_A^i \rangle \cdot P_B^i ~.
\end{equation}
We then get
\begin{eqnarray}
\label{eqnSequivSum}
\| S \|_{X^* \rightarrow 2} &=& \sup_{ \| G_A \|_{X^*} \leq 1 } \| S(G_A) \|_2 \nonumber \\
&=& \sup_{ \| G_A \|_{X^*} \leq 1 } \left\| \sum_{i} e_i \cdot \langle G_A, P_A^i \rangle \right\|_2 \nonumber \\
&=& \sup_{ \| G_A \|_{X^*} \leq 1 } \left( \sum_{i} | \langle G_A, P_A^i \rangle |^2 \right)^{1/2}  \nonumber \\
&=& w_2(P_A^i ; X) ~.
\end{eqnarray}
On the other hand, using the duality relation between norms, we have
\begin{eqnarray}
\| R \|_{2 \rightarrow Y} &=& \sup_{ \| \lambda \|_{2} \leq 1 } \| R(\lambda) \|_Y \nonumber \\
&=& \sup_{ \| \lambda \|_{2} \leq 1 } \left\| \sum_{i=1}^{n} \langle \lambda, e_i \rangle \cdot P_{B}^i \right\|_Y \nonumber \\
&=& \sup_{ \| \lambda \|_{2} \leq 1 } \sup_{ \| G_B \|_{Y^*} \leq 1} \left| \left\langle G_B, \sum_{i=1}^{n} \langle \lambda, e_i \rangle \cdot P_{B}^i \right\rangle \right| \nonumber \\
&=& \sup_{ \| G_B \|_{Y^*} \leq 1} \sup_{ \| \lambda \|_{2} \leq 1 }  \left| \sum_{i=1}^{n} \langle \lambda, e_i \rangle \cdot \langle G_B, P_{B}^i \rangle \right| ~. \nonumber
\end{eqnarray}
By setting $\langle \mu, e_i \rangle := \langle G_B, P_B^i \rangle$, with $\mu \in \ell_2^n$, and using that $\ell_2^n$ is self dual, we get
\begin{eqnarray}
\label{eqnReqivSum}
\| R \|_{2 \rightarrow Y} &=& \sup_{ \| G_B \|_{Y^*} \leq 1} \sup_{ \| \lambda \|_{2} \leq 1 }  \left| \sum_{i=1}^{n} \langle \lambda, e_i \rangle \cdot \langle G_B, P_{B}^i \rangle \right| \nonumber \\
&=& \sup_{ \| G_B \|_{Y^*} \leq 1} \sup_{ \| \lambda \|_{2} \leq 1 }  | \langle \lambda, \mu \rangle | \nonumber \\
&=& \sup_{ \| G_B \|_{Y^*} \leq 1} \| \mu \|_2 \nonumber \\
&=& \sup_{ \| G_B \|_{Y^*} \leq 1} \left( \sum_i | \langle \mu, e_i \rangle |^2 \right)^{1/2} \nonumber \\
&=& \sup_{ \| G_B \|_{Y^*} \leq 1} \left( \sum_i | \langle G_B, P_B^i \rangle |^2 \right)^{1/2} \nonumber \\
&=& w_2(P_B^i ; Y) ~.
\end{eqnarray}
which finishes the first part of the proof.

Let us now show that the right-hand side of  (\ref{eqnEqivNormGamma2}) is smaller or equal to the left-hand side.
Let $\hat{P} = R \cdot S$ be the optimal factorization of $\hat{P}$ on the left-hand side of  (\ref{eqnEqivNormGamma2}).
Then there exist $P_A^i \in X$ and $P_B^i \in Y$ such that $R(\lambda) = \sum_{i=1}^{n} \langle \lambda, e_i \rangle \cdot P_{B}^i$ and $S(G_A) = \sum_{i=1}^n e_i \cdot \langle G_A, P_A^i \rangle$, respectively. Hence, $\sum_i P_A^i \otimes P_B^i$ is a valid representation of $P$ (see also (\ref{eqnRepP}) and (\ref{eqnEqivRSP})). Using (\ref{eqnSequivSum}) and (\ref{eqnReqivSum}) finishes the proof.

\bibliographystyle{alpha}
\bibliography{HilbertianTensorNormV3}

\end{document}